\documentclass[a4paper, 11pt]{article} 

\usepackage[utf8]{inputenc}
\usepackage[T1]{fontenc}
\usepackage{amsmath}
\usepackage{amssymb}
\usepackage{amsthm} 
\usepackage{tikz-cd}
\usepackage{enumerate} 

\usepackage{hyperref} 
\usepackage{enumerate} 
\usepackage{graphicx}  
\usepackage{subcaption}
\usepackage{bm}
\usepackage{url}

\usepackage{booktabs}

\usepackage[maxbibnames=6]{biblatex}
\usepackage{algorithm2e}
\RestyleAlgo{ruled}

\usepackage{amsfonts}

\newcommand{\overbar}[1]{\mkern 1.5mu\overline{\mkern-1.5mu#1\mkern-1.5mu}\mkern 1.5mu}

\newtheorem{thm}{Theorem}

\newtheorem{lem}{Lemma}

\newtheorem{prop}{Proposition}

\DeclareMathOperator{\amin}{Amin}
\DeclareMathOperator{\conv}{Conv}
\DeclareMathOperator{\ner}{Nrv}
\DeclareMathOperator{\geo}{Geo}

\DeclareMathOperator{\Max}{Max}
\DeclareMathOperator{\ver}{Ver}

\newcommand{\vect}[1]{\boldsymbol{\mathbf{#1}}} 
\newcommand{\m}{\mathcal}
\newcommand{\outd}{{\textnormal{out}}}
\newcommand{\ind}{{\textnormal{in}}}

\usepackage{soul}

\usepackage{comment}
\usepackage{geometry}
\usepackage{authblk}  

\title{Circular Directional Flow Decomposition of Networks}

\author{Marc Homs-Dones, Robert S. MacKay, Bazil Sansom, and Yijie Zhou}

\date{\today}

\begin{document}

\maketitle

\begin{abstract}
We introduce the Circular Directional Flow Decomposition (CDFD), a new framework for analyzing circularity in weighted directed networks. CDFD separates flow into two components: a circular (divergence-free) component and an acyclic component that carries all nett directional flow. This yields a normalized circularity index between 0 (fully acyclic) and 1 (for networks formed solely by the superposition of cycles), with the complement measuring directionality. This index captures the proportion of flow involved in cycles, and admits a range of interpretations --- such as system closure, feedback, weighted strong connectivity, structural redundancy, or inefficiency. Although the decomposition is generally non-unique, we show that the set of all decompositions forms a well-structured geometric space with favourable topological properties. Within this space, we highlight two benchmark decompositions aligned with distinct analytical goals: the maximum circularity solution, which minimizes nett flow, and the Balanced Flow Forwarding (BFF) solution, a unique, locally computable decomposition that distributes circular flow across all feasible cycles in proportion to the original network structure. We demonstrate the interpretive value and computational tractability of both decompositions on synthetic and empirical networks. They outperform existing circularity metrics in detecting meaningful structural variation. The decomposition also enables structural analysis --- such as mapping the distribution of cyclic flow --- and supports practical applications that require explicit flow allocation or routing, including multilateral netting and efficient transport.
\end{abstract}

\medskip


\section{Introduction}

Many complex systems can be naturally represented as networks, where nodes represent system components, and edges capture weighted interactions between them. Very often, these interactions are inherently directional:~a flow from $A$ to $B$ (such as a payment, message, or resource transfer) is fundamentally distinct from a flow in the opposite direction. Such relationships are represented by weighted directed edges, which encode the magnitude and direction of flow between nodes \cite{Bang_2009}. Yet beyond this local structure, a deeper question arises:~how much of the system’s overall flow follows directed, source-to-sink pathways, and how much participates in closed-loop structure within the network \cite{how_directed, Johnson_2020}? The presence and extent of circularity in the sense of closed-loop flow, is a structural property with diverse meanings and implications across a broad array of different domains.

High levels of circular flow are often associated with \textit{efficiency}. In ecosystems, mature systems retain and recycle nutrients more effectively, reducing reliance on external inputs and minimizing losses \cite{Odum, Vitousek, Lajtha_2020}.\footnote{This internal cycling is often disrupted by disturbance or degradation \cite{Bormann_1968, Göttlein_2023}.} In logistics, efficient round-trip routing reduces ``empty loads'', improving capacity use and lowering costs \cite{Fisher_1995}. In payment systems, recirculating liquidity reduces reliance on external funding by increasing the velocity of money \cite{Hayakawa_2020, Byck_2020, Duca-radu_2021}. Similarly, circular economy models in production and manufacturing promote material efficiency by closing loops through reuse and recycling \cite{Layton_2016, Loon_2023}.

In other contexts, circularity can signal \textit{redundancy} or \textit{inefficiency}. In over-the-counter derivatives markets, cycles of exposures add counterparty-risk without affecting nett market positions, and can be eliminated through multilateral netting \cite{DErrico_2021}.\footnote{``Netting'' means reduction of obligations in a network by selecting a cycle of obligations and an amount up to the minimum of the obligations in the cycle and then subtracting that amount from each of the obligations in the cycle; the resulting obligations achieve the same nett obligations.} In transport and communication systems, unnecessary cycles increase costs and resource use --- such as duplicated messages in peer-to-peer networks that consume bandwidth and create processing bottlenecks \cite{Zhu_2008}. In biological systems, futile metabolic cycles occur when enzymes drive opposing reactions simultaneously, leading to energy dissipation without productive output \cite{Katz_1978}, offering targets for metabolic engineering to improve, e.g.~energy and carbon-use efficiency \cite{Amthor_2019, Garcia_2023,Barratt_2009,Cairns_2004,Morandini_2003}.

Circular flow can also embed \textit{positive feedback}, enabling self-reinforcing dynamics, or serve as a conduit for \textit{connectivity}, allowing information, stress, or contagion to propagate through a system. In finance, tightly coupled balance sheets mean circular links can both amplify and relay shocks across institutions \cite{Bardoscia_2017, Sansom_2021, Silva_2017}. In epidemiology, feedback loops in contact networks help sustain endemic disease states  \cite{Drysdale_2025,Klaise_2016}. 

Conversely, circular structure can enhance \textit{stability} and \textit{regulation} by embedding negative feedback or buffering external fluctuations. Ecological nutrient cycles, for instance, may dampen perturbations in resource inputs \cite{Loreau_1994}. The same “detour-and-release” logic appears in metabolic regulation, where biochemical cycles temporarily sequester flux and then release it to prevent concentration spikes and keep downstream reactions within operational ranges \cite{Misselbeck_2019, Qian_2006, Tolla_2015}. In computer science an analogous mechanism is exploited by routing protocols deliberately sending packets into detour cycles, smoothing burst traffic and extending mesh-network lifetime \cite{Liao_2023}. In systems with intransitive (cyclic) competition, no single strategy can dominate all others (the rock-paper-scissors pattern). Instead, the cycle of local victories yields dynamic coexistence and maintains biodiversity in microbial, plant, and animal communities \cite{Sinervo_1996, Reichenbach_2007, Allesina_2011}.\footnote{The degree of intransitivity often determines whether a community converges to monoculture or persists in perpetual motion \cite{Gallien_2017}, and similar cyclic dominance motifs have been documented in economics, politics, and social choice theory \cite{Klimenko_2015, Strang_2022}.}

Circular connectivity may also reflect \textit{community} or \textit{structural embedding} --- such as economic integration \cite{Sada_2021, Kichikawa2019}, self-containment in localised production-consumption loops \cite{Braun_2021}, or economic embeddedness within mutual credit systems \cite{Mattsson_2023, Sardex, Kelly_2014}.

Given circularity’s importance across such diverse domains, a range of methods have been taken to identify and quantify circularity in directed networks. These include topological approaches, such as counting elementary cycles or analysing their length distributions \cite{Mattsson_2023}, and spectral methods based on properties of the adjacency matrix, Laplacian or modified Laplacian \cite{Haruna_2016,Kichikawa2019,Finn_1976,Iskrzyński_2021,Layton_2017}.

However, existing approaches suffer from several limitations. Some are not genuine measures of circular structure, relying on vector field decompositions \cite{Haruna_2016,Kichikawa2019} that violate flow conservation or fail to respect network constraints such as edge capacities and direction. Others ignore edge weights, as in topological cycle counts \cite{Mattsson_2023} or certain unweighted spectral methods \cite{Layton_2017}. Some methods reduce circularity to a single scalar index, offering no information about how circular flow is distributed across the network \cite{Finn_1976, Layton_2017}. Flow decomposition-based approaches provide a natural framework for capturing circularity not only at the overall network level, but also individual flows. However existing formulations suffer from conceptual and methodological shortcomings:~some violate basic physical constraints such as flow conservation or capacity limits \cite{Haruna_2016, Kichikawa2019}, while others produce only a single, somewhat arbitrary decomposition, without acknowledging the underlying space of feasible decompositions, or offering principled criteria for selecting among them \cite{Ulanowicz_1983}. In some cases, computational complexity and scalability constraints also significantly limit the applicability of these methods \cite{Ulanowicz_1983, Mattsson_2023}.

In this paper, we introduce a framework for analyzing circularity in weighted directed networks, based on the decomposition of flow into two complementary components:~a circular component with zero nett flow and an acyclic component that accounts for all directional (nett) flow across the network. This decomposition respects flow conservation, adheres to edge capacity constraints, and operates directly in the native units of flow --- whether dollars, carbon atoms, passenger journeys, or otherwise. The proportion of total flow allocated to the circular component provides a natural, dimensionless measure of circularity, ranging from 0 (fully directional) to 1 (pure circulation).

Crucially, the CDFD is not unique. We show that the set of all valid decompositions forms a well-structured solution space:~a contractible polytope complex capturing every feasible way to partition flow between cyclic and acyclic components. 

Within this space, we identify two decompositions of particular interest. The \textit{Maximum Circularity Decomposition} identifies the solution with the largest possible circular component (equivalently, the smallest nett directional flow). The \textit{Balanced Flow Forwarding (BFF) Decomposition}, is a uniquely defined solution that allocates circular flow proportionally, based on local structure and information.

We argue the relevance of these decompositions across a wide range of applications. They are computationally tractable even for large-scale networks and support both structural analysis, such as quantifying circularity or identifying feedback-rich sub-networks; and constructive tasks where flow must be explicitly allocated. Practical use cases include multilateral netting and routing problems, where the local nature of the BFF decomposition could allow distributed or decentralized implementation. Empirical comparisons with other measures highlight the advantages of our framework in capturing structurally meaningful and interpretable measure of circularity.

It is also important to highlight that, whilst in this paper we largely focus on the circular part of the decompositions, one can analogously focus on the directional part of the decomposition. Taking 1 minus the circularity provides a natural measure of ``directedness'' that may be valuable in contexts where the asymmetry or irreversibility of flows carries structural or functional significance. 

The paper is organised as follows. We introduce the framework of circular directional flow decompositions in section~\ref{sec:intro_decomposition}. We motivate it and compare it with other quantifications of circularity in section~\ref{sec:motivation_comparison}. Our main results on the decomposition space are given in section~\ref{sec:decomp_space}, with proofs in a set of appendices. In section~\ref{sec:choosing_rep}, we highlight two particular choices of decomposition. Results of testing our methods on example networks, both simulated and real datasets, are given in section~\ref{sec:numerics}. The paper ends with conclusions in section~\ref{sec:conclusions}.

\section{Framework}
\label{sec:intro_decomposition}

We consider weighted directed networks (also known as weighted digraphs) with finite set of nodes $\mathcal N$ and set of directed edges $\mathcal E$. We suppose there is at most one edge from node $i$ to node $j$,  denoted by $ij$, which carries a weight $w_{ij}>0$. We record all the weights in the weighted adjacency matrix, $\mathbf w$, where $w_{ij} = 0$  indicates the absence of the edge $ij$. We will identify weighted adjacency matrices with the network they represent as they carry the same information. Self-loops $ii$ are permitted and multiple edges from $i$ to $j$ can be amalgamated into a single edge by adding their weights.   Note that edges from $j$ to $i$ are kept distinct from those from $i$ to $j$.
For unweighted datasets we simply set all weights to one. 
For each node $i$ we define its in-weight and out-weight by 
\begin{equation*}
 w_{i}^\ind: = \sum_{j\in \mathcal N_{\mathbf w}}w_{ji} 
\hspace{1cm}\textnormal{and}\hspace{1cm}
 w_{i}^\outd:=\sum_{j\in \mathcal N_{\mathbf w}}w_{ij}.
\end{equation*}

Before presenting the circular plus directional flow decomposition (CDFD) we rigorously define the individual terms. We will use the term flow interchangeably with network, unlike some authors in computer science who impose stronger restrictions on flows, as it gives good intuition in our context. A directed cycle in a network, or for us simply a \emph{cycle},  is a finite sequence of distinct nodes (except the two ends that must coincide) with consecutive nodes  connected by following directed edges.
A network is acyclic, also known in the literature as DAG (directed acyclic graph),  if it does not contain any cycles. We will refer to acyclic networks as \emph{directional}, as their flows have a global direction. More precisely they can be characterised by having a node ordering for which all edge flows go strictly downwards, see appendix~\ref{app:minimu_cost_equivalence}. 
At the other extreme, we say a network is \emph{circular} (also known as balanced) if for each node the in- and out-weight coincide.  A special case is a \emph{balanced cycle} consisting of a directed cycle with equal weights on its edges. We note balanced cycles are the natural generalisation of cycles in a weighted flow context as they do not contain bottlenecks.   Any circular network can be expressed as the sum of balanced cycles \cite{cycle_decompositions} (although in general there will be many such decompositions) 
justifying the terminology.

Given a network $\mathbf w$, we say that the pair $(\mathbf c, \mathbf d)$ is a CDFD of $\mathbf w$ if $\mathbf c$ is circular, $\mathbf d$ is directional, and both are subnetworks of $\mathbf w$ such that their sum equals $\mathbf w$. Formally, this means $0\leq c_{ij}, d_{ij}\leq w_{ij}$ for all $ij\in \m E_{\mathbf w}$, and $\mathbf c +\mathbf d = \mathbf w$.
Then, we define the circularity  of the decomposition to be simply the ratio between the total weight in the circular part and in the original network, that is,  
\begin{equation}
    \frac{\sum_{ij\in \mathcal E_{\mathbf c}}c_{ij}}{\sum_{ij\in \mathcal E_{\mathbf w}}w_{ij}},
    \label{eq:circularity_ratio}
\end{equation}
which ranges from 0 to 1. The directionality is defined analogously and its sum with circularity is 1 as  $\mathbf c + \mathbf d = \mathbf w$. Thus it does not give extra information and we will focus on the circularity in this work. However, it is important to keep in mind that in some applications the directionality may be of more interest. Furthermore, while developing the theory it will be convenient to concentrate on  the directional part of CDFDs.

\begin{figure}[ht]
    \centering
\includegraphics[width=0.85\textwidth]{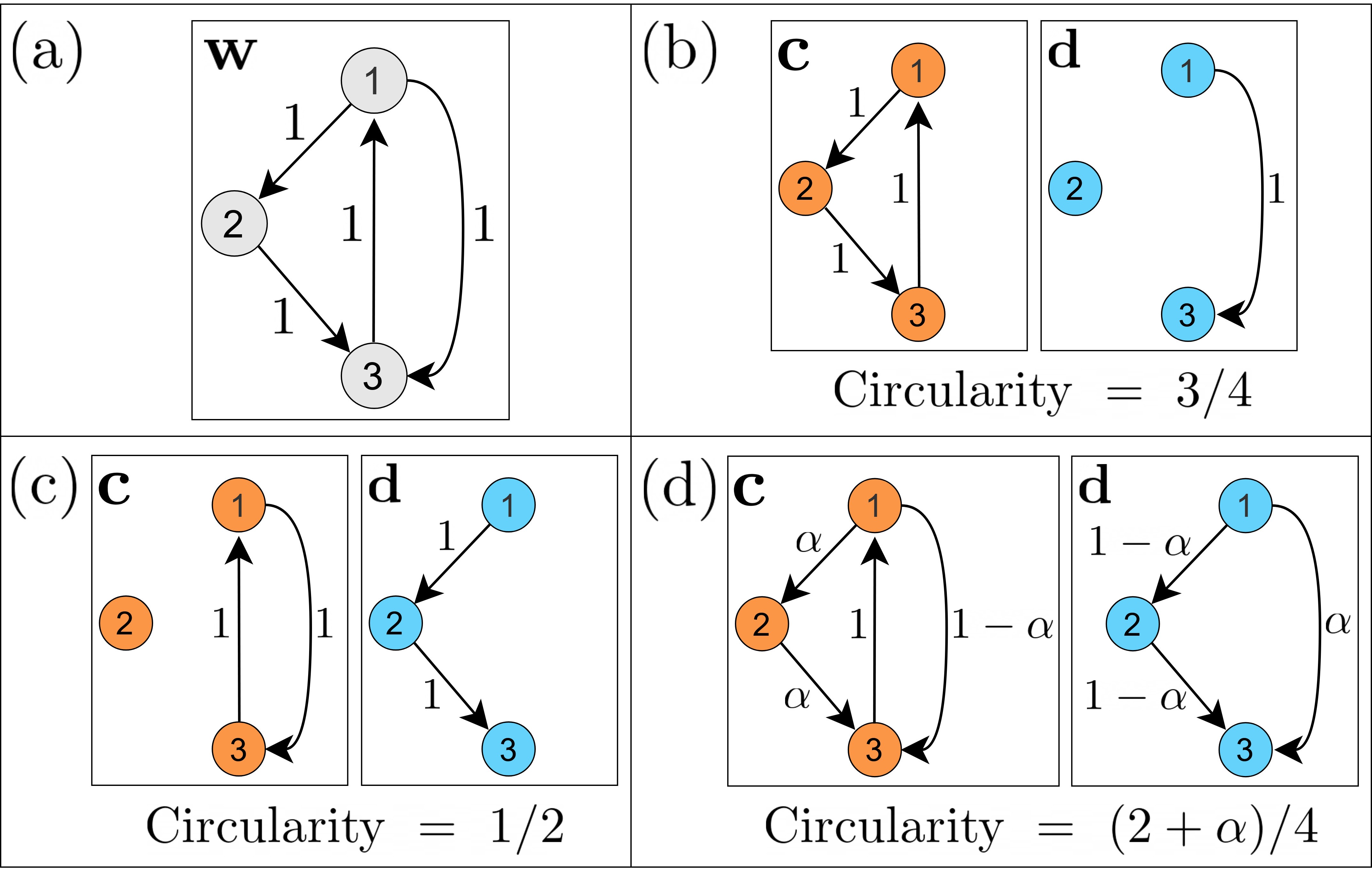}
\caption{Different decompositions of the network depicted in (a). In (d) we have  $0\leq \alpha \leq 1$ so that all weights are non-negative. }
    \label{fig:first_eg}
\end{figure}

Note that we have defined circularity with respect to a particular decomposition, as the decomposition is generally not unique. For example, consider the network shown in figure \ref{fig:first_eg}(a). If we take the circular part to be the left (respectively, right) balanced cycle with maximum capacity, we obtain the decomposition in figure \ref{fig:first_eg}(b) (respectively, figure \ref{fig:first_eg}(c)). The circularity values for these decompositions differ, since the left cycle is longer and thus carries more total weight. Now consider a convex combination of the two decompositions in (b) and (c), as illustrated in figure \ref{fig:first_eg}(d). The circular part is no longer a single balanced cycle but rather the sum of two balanced cycles with edge weights $\alpha$ and $1-\alpha$ respectively. Notably, the remaining flow is acyclic, as seen on the right of figure \ref{fig:first_eg}(d), where all edge flows are directed downward. Thus, for this example, we have infinitely many decompositions, ranging between the extremes of figure~\ref{fig:first_eg}(b) and (c).

A crucial concept in our work is the  \emph{decomposition space} of a network $\mathbf w$, which is simply the set of all its possible decompositions endowed with the subspace topology in $\mathbb R^{|\mathcal E_{\mathbf w}|}$.
We can then summarise the previous paragraph by saying that the decomposition space of figure~\ref{fig:first_eg}(a) is a segment between the decompositions in figure~\ref{fig:first_eg}(b) and (c).  
It would be cumbersome to always keep track of both the circular and directional part of decompositions, so in the upcoming sections we will work with the directional\footnote{We keep track of $\mathbf d$ instead of $\mathbf c$ despite being mainly interested in circularity, as $\mathbf d$ is directly related to the minimum cost flow problem, see section~\ref{sec:minimum_cost}.}
part, and the circular one can be recovered by $\mathbf c = \mathbf w -\mathbf d$. We will slightly abuse notation and say that $\mathbf d$ is a directional part of $\mathbf w$,
and refer to 
\[D=\left \{\mathbf d \; | \; (\mathbf w-\mathbf d,\mathbf d) \textnormal{ is a CDFD of }\mathbf w \right \}.\] 
as the decomposition space. The \emph{circularity range} is then defined as the set of circularity values \eqref{eq:circularity_ratio} of all these decompositions.  

Note that the extreme  case of the circularity range containing 0 or 1 can only occur when $\mathbf w$ is already directional or circular. 
What is more, if $\mathbf w$ is directional its circular part must be null so  $D = \{ \mathbf w \}$ and the only value in the circularity range is 0. Analogously, if $\mathbf w$ is circular $\mathbf d = \mathbf w - \mathbf c$ is also circular, so $\mathbf d = \mathbf 0$ and $D = \{ \mathbf 0 \}$ with circularity 1.

\section{Comparison with other circularity measures}
\label{sec:motivation_comparison}

At first glance the introduction of the decomposition space to measure circularity may seem contrived. However, it arises naturally when considering the possible balanced flows on a network $\mathbf w$ where the weights $w_{ij}$ represent the maximum capacity from $i$ to $j$. To see this we will first consider other possible definitions of circularity.

One intuitive way to quantify circularity in a directed network is to count the number of directed cycles, an idea adopted by several recent studies such as \cite{Mattsson_2023, Sardex, Harang_2020}. Despite its apparent simplicity, this approach has some significant limitations. 

First, raw cycle counts are hard to interpret. To assess whether a network has ``many'' or ``few'' cycles, a suitable null model is needed --- yet defining such a model is often nontrivial. Moreover, the number of cycles in a network can increase super-exponentially with the number of nodes \cite{Dominguez_2014}, so implementations typically restrict attention to short cycles (e.g., of length $\leq$ 4 or 5) or a handful of predefined motifs \cite{Milo_2002}. This truncation cannot yield a true global circularity measure.

Also importantly, this method ignores edge weights, offering no insight into the volume or intensity of flow. While cycle counts may suggest topological potential for circularity --- such as the richness of feedback paths --- they say little about actual flow patterns.\footnote{For instance, whilst \cite{Harang_2020} are very clear that “strength of connections in cycles” is crucial for the neural dynamics they study, they lack a measure to formalise that strength into a weighted cyclicity index, and rely on a cycle-counting based cyclicity index.} In systems where flow is conserved and constrained (such as metabolic, financial, transport, or data networks), circularity depends not just on the existence of cycles but on how much flow actually moves, or could move, through them. 
A network may have many cycles yet be functionally dominated by directional (acyclic) flow, see figure~\ref{fig:fake_circularity}(b) for a minimal example. Conversely, even a network with few cycles can exhibit high circularity if a large share of total flow recirculates. A meaningful measure of circularity must therefore account for both the presence of cycles and the distribution and magnitude of flow.

One natural extension is to sum the flow that could circulate in each cycle, considering capacity constraints. This captures flow intensity --- but overestimates circularity when cycles share edges. For example, in figure \ref{fig:first_eg}(a), two balanced cycles overlap on one edge, so they cannot both circulate at full capacity without exceeding that edge's limit. Naively summing their flow would double-count. In systems governed by conservation laws, such double-counting is not physically meaningful. 
A proper measure must therefore allocate edge capacities consistently across overlapping cycles.

This calls for a cycle-aware flow decomposition, which explicitly tracks how shared edges are used by different cycles. Our proposed CDFD does exactly this:~it separates total flow into circular and acyclic components, respecting both flow conservation and edge-level capacity constraints (section \ref{sec:intro_decomposition}). 

Once such a decomposition is obtained, a natural measure of circularity is the ratio of circular to total flow (\ref{eq:circularity_ratio}) ($0$ for an acyclic network and $1$ for a pure circulation). This measure is physically grounded and interpretable, requiring no null model. It allows for meaningful comparisons across networks and domains. Moreover, since the decomposition respects conserved quantities, its components are directly interpretable (e.g., ``30\% of flow is circular, equivalent to \$12 million'', or ``10,000 passenger journeys'').

This CDFD measure captures the total volume of flow participating in circular structures --- a “how much water is in the pipes” perspective.\footnote{Assuming the weights represent flow rates in the pipes, and all the pipes have the same length.}
An alternative is to define a circular flow rate, where the circular component is decomposed into balanced elementary cycles and one edge-weight is counted per cycle. This approach would be invariant to edge-splitting and treat short and long cycles equally, which may be desirable in certain contexts. However, it depends not only on the initial CDFD but also on the specific choice of cycle decomposition of the circular component. Questions of normalization also arise --- though natural candidates exist. We leave these extensions for future work.

Below we discuss several widely used circularity metrics, highlighting how they differ from our CDFD approach in their assumptions, computational demands, and interpretability.

The network Helmholtz–Hodge decomposition (HHD) splits any network flow into a diverg\-ence‑free (circular) and a gradient (directional) component \cite{Johnson_2013,Strang_2020} with respect to an inner product that is usually taken to be the standard one or the one weighted by the reciprocals of the flows, though there is a huge space of alternatives. The share of circular flow (between 0 and 1) has been proposed \cite{Haruna_2016} and widely used as a “circularity” score, and the decomposition has been used for structural analyses of circulation 
in a range of settings from production‑network community detection \cite{Kichikawa2019} to brain‑graph comparison \cite{Anand_2024}.  The non-uniqueness coming from the choice of inner product is rarely mentioned, but is a serious issue, making a whole space of HHDs.  For us, this makes it reasonable to consider the entire space of CDFDs.

 As shown by MacKay et al.~\cite{how_directed}, however, rather than circularity this score in fact measures trophic incoherence:~the departure from a perfectly layered hierarchy \cite{Johnson_2014,how_directed}. A fully circular flow scores 1, but purely acyclic networks can still register positive values (see figure \ref{fig:fake_circularity}(a)) because the metric is driven by hierarchical disorder, so not only cycles but also feed-forward structure contribute to incoherence, and only a perfectly layered hierarchy achieves a score of 0. Hence the measure provides only a proxy for circularity, which is unable to distinguish between low levels of circularity vs feed-forward structure. 

HHD also reorients edges and ignores capacities, creating fictitious loops even in DAGs (see figure \ref{fig:fake_circularity}(a)). Such reversals are unacceptable in applications where flow direction and conservation are essential, such as financial netting and gridlock resolution \cite{Bech_2002,DErrico_2021}, or metabolic and biomass flux analysis \cite{Chapman_2017,Kritz_2010,Iskrzyński_2021}.

\begin{figure}[ht]
    \centering
\includegraphics[width=0.85\textwidth]{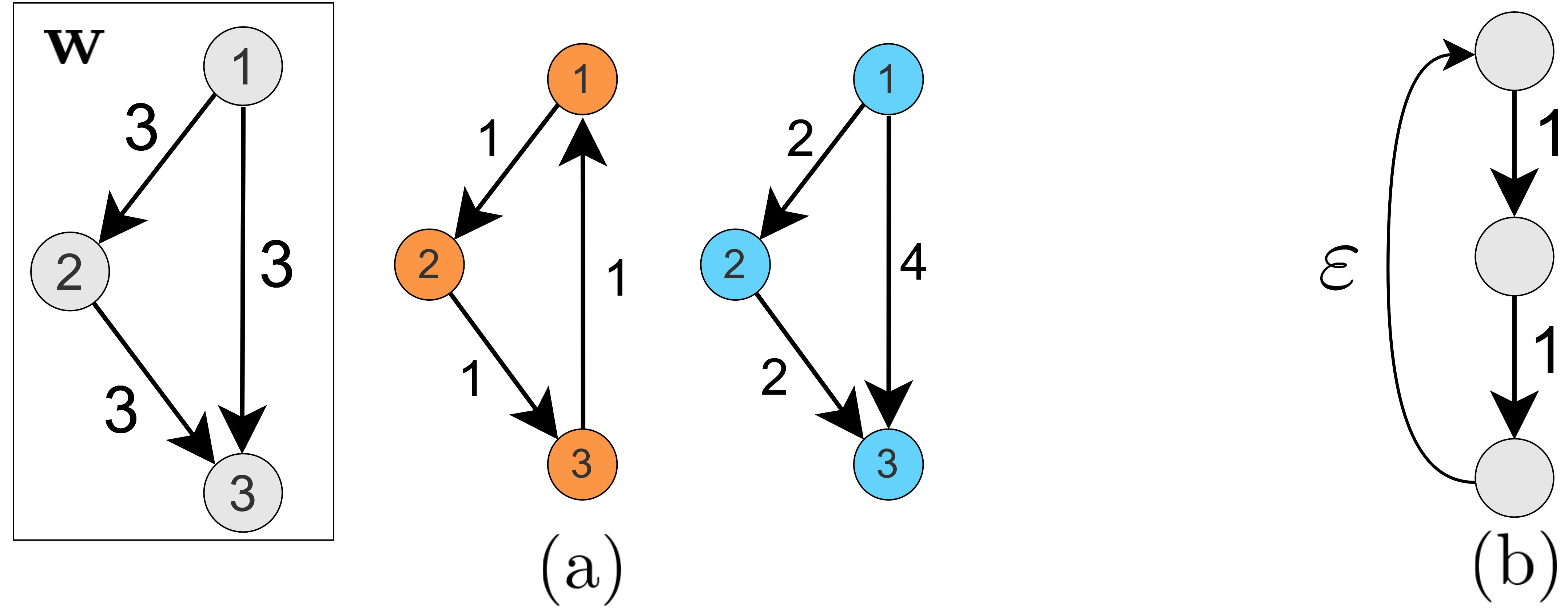}
\caption{Examples of problematic networks for other measures of circularity. In (a)  an HHD decomposition of an acyclic feed-forward motif $\mathbf w$ is shown. As HHD can reverse edge direction, the circular part is non-null and the HHD circularity is non-zero. In (b), the network is topologically very circular but is dominated by downwards flow when $\varepsilon$ is small. LM circularity incorrectly identifies it as fully circular.  }
    \label{fig:fake_circularity}
\end{figure}

Unlike HHD, our CDFD satisfies the natural baseline, 0 for any acyclic network, 1 for a pure cycle. Moreover, because the CDFD respects flow conservation and edge limits, it provides a sound foundation for structural analysis, as well as for applied flow allocation problems (such as netting, routing, rebalancing) where the raw HHD provides no practical utility.

A variant of HHD was introduced in \cite{Moutsinas2021}, which distinguishes between upstream and downstream directedness and thereby gives two notions of circularity.  But it suffers from the same problems as HHD.

A more straightforward method than HHD and CDFD is simply to compute the fraction of total flow carried by all edges that participate in a cycle, giving a 0--1 score we will call LM circularity\footnote{After Luo and Magee \cite{Luo_2011}} (the complement of this measure has been proposed for quantifying the extent to which local flows follow an overall ``underlying direction'' \cite{Luo_2011}). 
The simplicity of this measure is appealing, but it treats all flow on edges participating in cycles as circular, even if most of it runs one‑way from source to sink, so the metric systematically inflates circularity (see figure \ref{fig:fake_circularity}(b)).

In strongly‑connected graphs, where every edge lies on a cycle, LM circularity is always 1, no matter how directional the actual flows may be. Even in partly connected graphs, adding weight to a single edge in a cycle will increase LM circularity, even if it increases nett flow on the cycle. Consequently the measure can signal high circularity in networks with little true circulation, making it an unreliable indicator.

Closely related to our work, Ulanowicz \cite{Ulanowicz_1983} introduced a method to decompose steady-state ecological networks into two components:~a purely cyclic component (closed feedback loops) and a residual acyclic component representing once-through “trophic” flow. This approach has also been applied outside of ecology in other domains, such in the analysis of metabolic networks \cite{Kritz_2010, Chapman_2017} and trade networks \cite{Iskrzyński_2021}. Although Ulanowicz’s original aim was structural interpretation, a scalar circularity measure (the ratio of circular to total flow) can be readily derived from this decomposition. In fact Ulanowicz’s scheme yields a CDFD --- hereafter the Ulanowicz decomposition. 

The Ulanowicz Decomposition is defined through a heuristic procedure:~all simple cycles are first enumerated. Iteratively, the minimum‑weight arc is chosen; and its flow is removed from every cycle that uses it, using a heuristic “circulation‐probability” split. Two arbitrary choices therefore shape the outcome: (a) the order in which arcs are selected, and (b) the heuristic for sharing flow among overlapping cycles. Complete cycle enumeration also makes the algorithm scale poorly, curbing its use, e.g.~in genome‑scale metabolic models despite early promise \cite{Kritz_2010}.

Once we view the Ulanowicz method as giving one point in a larger CDFD solution space, its ad‑hoc choices lack a principled rationale. The alternative decompositions we introduce (Section \ref{sec:choosing_rep}) avoid full cycle enumeration, run in polynomial time, and are tailored to specific analytic goals --- whether quantifying circularity, exposing feedback structure, or optimising netting or routing. Ulanowicz’s “circulation probabilities” resemble our Balanced‑Flow Forwarding (BFF) concept (Section \ref{sec:BFF}), but BFF is derived from global flow dynamics, requires no edge ordering, and is far more efficient.

Another approach, widely used in ecology and other fields,\footnote{Some examples are \cite{Allesina_2004,Iskrzyński_2021,Braun_2021,Loon_2023,Jonge_2022,Layton_2016,Kazanci_2023}.} is the Finn Cycling Index (FCI) \cite{Finn_1976,Finn_1980} which reports the fraction of total throughput that “loops back” into the same nodes. Cast the flow network as an absorbing Markov chain whose transition matrix is the normalised flow matrix $ G $. The diagonal of $\mathbf (I-G)^{-1}$ gives the expected number of visits to each node --- including returns --- by a “random walker” before exiting.\footnote{Where $I$ is the identity matrix} FCI uses this to compute the share of global throughput accounted for by such revisitations.

Whilst intuitive, the FCI is path‑blind:~it says nothing about which edges carry the recurrent flow. It also flattens out as a system approaches closure, the measure loses discrimination, and $\mathbf (I-G)$ becomes ill‑conditioned. Our CDFD instead partitions each edge into directional and cyclic components, and works the same for open systems or closed steady states. For many applications the two indices can be complementary.

We note that spectral approaches are also taken to quantifying circularity, with some studies using the max eigenvalue of the unweighted adjacency matrix \cite{Layton_2016,Layton_2017,Fath_2007,Morris_2021,Borrett_2007}. Of course this measure is path-blind (no decomposition), weight blind (purely topological), but also insensitive to all variation in circularity outside the leading strongly connected component (SCC), unbounded, and its meaning is sensitive to network size and mean-degree, making raw values not in general meaningfully comparable across networks. Riane \cite{Riane_2025} extends the raw‑eigenvalue idea by adding weights and a node‑level decomposition but inherits the core spectral‑radius limitations.

We leave numerical comparisons between different definitions of circularity for section~\ref{sec:numerics_comparison}.

\section{Analysis of the decomposition space}
\label{sec:decomp_space}

\subsection{Minimum cost flow formulation}
\label{sec:minimum_cost}

When thinking of the network as a flow, it is natural to consider the nett flow that each node receives, i.e.~the difference between the in- and out-weights. As the circular part of a decomposition is balanced, the nett flow for all nodes is null, so all the nett flow from $\mathbf w$ has to be in the directional part. More precisely, for a network  $\mathbf x$ to be the directional part of $\mathbf w$ it is needed that 
\begin{equation} x_{i}^{\textnormal{in}}- x_{i}^{\textnormal{out}}= w_{i}^{\textnormal{in}}- w_{i}^{\textnormal{out}}\hspace{0.5cm}\forall i\in\mathcal N_{\mathbf w}, \hspace{1cm}\textnormal{and}\hspace{1cm}0\leq x_{ij}\leq w_{ij} \hspace{0.5cm} \forall ij\in \mathcal E_{\mathbf w},
\label{eq:x_constrains}
\end{equation}
which guarantees that $\mathbf x$ is a subnetwork of $\mathbf w$ and  $\mathbf w-\mathbf x$ is circular. 
On top of this constraint, we need to check that $\mathbf x$ does not contain cycles. 
One approach is to require $\mathbf x$ to be minimal over the flows that ``transport'' the nett flows of $\mathbf w$ between nodes, as then it will not have redundant cyclic flows. This can be achieved by simply minimising the total flow of $\mathbf x$, or more generally, setting edge  costs $\kappa_{ij}>0$ and minimising the total cost, that is,
\begin{equation}\min_{\mathbf x \textnormal{ in \eqref{eq:x_constrains}} }\sum_{ij\in\mathcal E_{\mathbf w}}\kappa_{ij}x_{ij}.
\label{eq:min_x}
\end{equation}

Linear optimisation problems with constraints \eqref{eq:x_constrains} and objective function \eqref{eq:min_x}, are known as minimum cost flow problems.  We have shown that for positive costs, which we will assume from now on and abbreviate by $\vect \kappa>0$, their solutions are directional parts of $\mathbf w$. 
Crucially, a reciprocal statement also holds, which will be key for further exploration of the decomposition space. 
\begin{thm}
    Any solution of \eqref{eq:min_x} is a directional part of $\mathbf w$, and any directional part of $\mathbf w$ can be attained as the solution of \eqref{eq:min_x} for appropriate costs. That is,
    \begin{equation}
    D = \bigcup_{\vect \kappa>  0} \operatorname*{arg\,min}_{\mathbf x \textnormal{ in \eqref{eq:x_constrains}}} \sum_{ij\in \mathcal E_{\mathbf w}}\kappa_{ij}x_{ij}.
    \label{eq:arg_min=D}
    \end{equation}
    \label{thm:CDFD_eqi_min_flow}
\end{thm}
We prove this in appendix~\ref{app:minimu_cost_equivalence}.

\subsection{Topology of decomposition space}
\label{sec:topology_decomposition}

First, let us explore the decomposition space of figure~\ref{fig:first_eg}(a) using theorem~\ref{thm:CDFD_eqi_min_flow}. When choosing costs $\vect \kappa$ with
\begin{equation}
    \kappa_{12}+\kappa_{23}>\kappa_{13}
    \label{eq:first_eg_inequality}
\end{equation}
the solution to \eqref{eq:min_x} gives rise to figure~\ref{fig:first_eg}(b), as the downwards path on the right is cheaper than the one on the left, whereas figure~\ref{fig:first_eg}(c) is attained if the inequality is flipped.
If the costs give an equality in \eqref{eq:first_eg_inequality}, then all the decompositions are optimal as both downward paths have the same total cost. Then the argument-minimum in \eqref{eq:arg_min=D} for a specific cost gives us the whole decomposition space and thus $D$ is convex, see appendix~\ref{app:background_polytope} for background.
However, for some networks no single cost matrix gives rise to the whole decomposition space, see appendix~\ref{app:non-convex} for an example, in which case the decomposition space will not be convex.  
 
A priori, a non-convex decomposition space could have wildly different types of decompositions with no decomposition bridging the gap between them. Even worse, we might fear that there are networks with disconnected circularity ranges, which would be extremely counterintuitive and undesirable. Fortunately, neither of these occurs as the decomposition space is connected and thus the circularity range is an interval. 

Intuitively, $D$ is connected because  $\{\vect \kappa > 0\}$ is connected and then theorem~\ref{thm:CDFD_eqi_min_flow} allows us to make a path between any pair of decompositions. In fact, theorem~\ref{thm:CDFD_eqi_min_flow} allows us to push much stronger topological properties from $\{\vect \kappa > 0\}$ to $D$ as is presented in appendix~\ref{app:decomposition_space}. They can be summarised by
\begin{thm}
    The decomposition space is a contractible polytope complex. 
\end{thm}
A contractible space is one that can be continuously deformed within itself to a point. In particular, it is connected and it does not contain ``holes'' of any dimension. Polytopes are geometric objects that generalise convex polygons to higher dimensions. Namely they are the convex hull of finitely many vertices in an affine space, for example segments, rectangles and pyramids. A polytope complex then is a collection of polytopes that intersect nicely. That is, their intersections are lower dimensional polytopes, such as vertex or segments.  For instance for the network presented in appendix~\ref{app:non-convex} the decomposition space consists of two segments joined at a vertex. A rigorous presentation of all these concepts can be found in appendix~\ref{app:background}.

\section{Choosing a decomposition }
\label{sec:choosing_rep}

In principle, all decompositions found are equally valid as they correctly identify circulating flow in the network.
Thus the interval of circularity values we obtain is simply a feature of the network, where each value represents different ways to view the flow in it. 
However, sometimes it is impractical to work with ranges and a specific value needs to be chosen.

In the following sections we propose two decompositions  for which we can then compute the circularity ratio by \eqref{eq:circularity_ratio}. The first method finds the maximum possible circularity in the given network, so it can also be used to find the upper end point of the circularity range. The second is an algorithm called balanced flow forwarding (BFF) that locally fairly distributes the circularity between different edges. With this we aim to find a ``middle'' point in the decomposition space, which may represent better a  typical circular part. Thus, we propose using BFF as the generic ``go to measure'' of circularity. However, specific contexts may suit the maximal one better (see section \ref{sec:choosing_rep} on the applied relevance of BFF vs.~max circularity) or even an entirely different  decomposition in the decomposition space.

\subsection{Maximal circularity}
\label{sec:maximum}

The maximal circularity of a network $\mathbf w$ is attained by finding a\footnote{Although such a subnetwork may not be unique, the circularity of all corresponding decompositions will coincide. When the specific decomposition is of interest, as in the upcoming section \ref{sec:compression}, the BFF decomposition has the advantage of being unique. }      balanced subnetwork $\mathbf c$ that maximises \eqref{eq:circularity_ratio}, or equivalently that   minimises  $\sum_{ij\in \mathcal E_{\mathbf d}}d_{ij}$ where $\mathbf d = \mathbf w- \mathbf c$.
This is precisely what the minimum cost problem \eqref{eq:min_x} does if we set all costs to 1. The  measure of directedness attained in this manner has already been proposed in the literature in a hierarchical approach \cite{Gupte_2011}, where attention was restricted to unweighted networks. This choice of decomposition is quite canonical, but it may give a skewed view of the circularity of the network, being an extreme. More precisely, the maximal circularity is achieved  at a vertex of the decomposition space (when it is achieved at a unique point), making it a quite specific type of decomposition. Moreover it is an end point of the circularity range, making it even more particular.

At the other end, one could consider minimal circularity, which together with the maximal would exhibit the whole range of circularity values. We will not consider this decomposition here as computing it is NP-hard, even for sparse unweighted networks, see appendix~\ref{app:minimum_circularity}. In contrast, there are many general purpose minimum cost flow algorithms,  suiting  different types of networks,
that run in polynomial time \cite[Page 395]{Ahuja_book}. 
For instance, the enhanced capacity scaling algorithm runs in $O((m \log n)(m+n\log n) )$ where $n =|\mathcal N_{\mathbf w}|$ and $m = |\mathcal E_{\mathbf w}| $. 
In our computations, see section~\ref{sec:numerics}, these algorithms performed better than our current implementation of BFF, which could incentivise the use of maximal circularity.

\subsection{Balanced Flow Forwarding (BFF)}
\label{sec:BFF}

The optimization approach seen in the previous section finds an extreme of the decomposition space. 
However, one may argue that a good representative of the decomposition space would be precisely the opposite, some sort of middle of it.  In this section we propose such a representative by having the flow in the circular part utilise all possible cycles, while locally attempting to preserve the proportions dictated by the weights on $\mathbf w$.  
 
Consider dynamics on the network where each node forwards the total flow received in the previous time step while preserving the flow proportion between edges. To avoid overflowing  edge capacities, we also limit the total forwarded flow in each node to the one forwarded in the previous time step. More formally, denote by $a_{ij}(t)$ the flow sent through edge $ij$ a time $t$, with starting flow $\mathbf a (0) = \mathbf w$. By preserving the flow proportions\footnote{By induction, keeping the proportions from the previous time step or keeping the initial proportions turn out to be equivalent.  } we mean that  
\begin{equation}
    \frac{a_{ij}(t)}{a^{\outd}_i(t)} = \frac{w_{ij}}{w_i^\outd} =: \bar{w}_{ij}
    \label{eq:ratio}
\end{equation}
for nodes with $a_{i}^\outd (t) \not =0$, or equivalently,
\begin{equation*}
    a_{ij}(t) =  a^{\outd}_i(t) \bar { w}_{ij}.
\end{equation*}
Then, the row out-weight vector $\mathbf {a}^\outd (t)$ completely determines the dynamics. The in-weights are given by
\begin{equation}
    a_i ^\ind (t)= \sum_{j\in\mathcal N_{\mathbf w}}   a_{ji}(t) = \sum_{j\in \mathcal N_{\mathbf w}}   a^{\outd}_j(t) \bar {w}_{ji} = [\mathbf {a}^\outd (t)\bar {\mathbf w}]_i ,
    \label{eq:a^in}
\end{equation}
where the last expression is a matrix multiplication. The updating rule for the out-weight flow is then  
\begin{equation}
    a^\outd_i (t+1) = \min(a^\ind_i (t), a^\outd_i (t)) = \min([\mathbf a^\outd (t) \bar {\mathbf w}]_i , a^\outd_i (t)).
    \label{eq:a(t+1)}
\end{equation} 
As $t \to \infty$, these dynamics converge to a balanced subnetwork, $\tilde{ \mathbf c}$, which also preserves the flow proportions. These properties follow from the fact that $\tilde {\mathbf c}^\outd$ is just a scaling of a stationary distribution of random walks in $\mathbf w$. We expand on this in appendix~\ref{app:BFF} where we show
\begin{thm}
    The vector $\tilde{\mathbf c}^\outd$ is the maximal invariant vector of $\bar{\mathbf w}$ bounded by $\mathbf w^{\outd}$. That is, $ \tilde{\mathbf c}^\outd \bar{\mathbf w} = \tilde{\mathbf c}^\outd $ and if $\mathbf  x\leq \mathbf w^{\outd}$ and $ \mathbf x \bar{\mathbf w}= \mathbf x$, then $\mathbf x\leq \tilde{\mathbf c}^\outd $; where inequalities are taken entrywise. 
    \label{thm:maximal_invariant}
\end{thm}
Unfortunately,  $\tilde{\mathbf c}$ is not a circular part of $\mathbf w$, as the rest of the flow may still have circularity in it. However, the maximality guarantees that the rest of the flow $\mathbf w-\tilde {\mathbf c}$ has a \emph{sink} node, that is a node with no outgoing edges. As sink nodes clearly play no role in circularity, we can remove them, and study the same dynamics on the resulting network. Recursively doing this leads to algorithm~\ref{alg:Comprsession_BFF} that finds a distributed circular part of $\mathbf w$. 

\begin{algorithm}[hbt!]
\caption{Balanced Flow Forwarding (BFF) algorithm. }
\label{alg:Comprsession_BFF}
\SetKwFunction{Fmaxinv}{Maximal\_Invariant}
\KwIn{ Weighted adjacency matrix of a graph $w \geq 0$}
\KwOut{Decomposition $(c,d)$ of $w$ where $c$ is the circular part and $d$ the directional one }
$g = w $\\
$c = 0$\\
\While{$g$ has a node}{
    \While{$\exists$ a sink $s$ of $g$}{
        Remove $s$ from $g$ and all edges incident to $s$
    }
    $\tilde c = $ \Fmaxinv{$g$}\\
    $c = c + \tilde c$\\
    $g = g - \tilde c$\\
}
$d = w-c$
\end{algorithm}

When adding networks in algorithm~\ref{alg:Comprsession_BFF}, we simply add the weights of matching edges, as one of them may be missing some nodes.  The  \texttt{Maximal\_Invariant} function finds the network preserving the flow proportions of $\mathbf g$ whose out-weight vector is maximally invariant as in theorem~\ref{thm:maximal_invariant} but  with respect to $\mathbf g$ instead of $\mathbf w$. 

Note that we are guaranteed to remove at least one sink node in each iteration (except on the first one), so algorithm~\ref{alg:Comprsession_BFF} converges in at most as many iterations as the number of nodes plus one.
Moreover, the resulting circular part $\mathbf c$ will have flow in all cycles of $\mathbf w$ as is shown in appendix~\ref{app:BFF}. Unfortunately, removing sinks,  which is necessary to get convergence to a decomposition,
does not preserve the flow proportions so neither does algorithm~\ref{alg:Comprsession_BFF}. Despite this, the algorithm  still adequately finds distributed decompositions. For instance, for the example in figure \ref{fig:first_eg}(a) the BFF decomposition is figure \ref{fig:first_eg}(d) with $\alpha = 1/2$ and its circular part has the same weights in each of its cycles. 
Note that in this example the decomposition space is a segment and the BFF decomposition is its midpoint.  A larger example is discussed in section \ref{sec:ex_decomposition_space}.
 
\subsubsection{Local vs non-local version}
The dynamics \eqref{eq:a(t+1)} proposed to find the maximal invariant vector in the previous section are local, i.e.~nodes only need to know about their own edges to execute their update. In practical applications, such as those discussed in section~\ref{sec:compression}, this is key, as it allows us to find a flow decomposition in a decentralised system while keeping edge information private. However, it comes at a cost, as \eqref{eq:a(t+1)} might not converge in finite time. When privacy is not a concern, it is much more efficient to directly compute the maximal invariant vector through standard linear algebra algorithms, which leads to a polynomial time algorithm for BFF, see appendix~\ref{app:BFF}.

In this non-local approach one can take an appropriate scaling of any stationary distribution instead of  \texttt{Maximal\_Invariant} in algorithm~\ref{alg:Comprsession_BFF}, as it will lead to the same decomposition, see appendix~\ref{app:BFF}. In turn, this strengthens the choice of the BFF decomposition as representative. Indeed, as shown in appendix~\ref{app:BFF},  any algorithm that at each iteration removes circular flow keeping the proportion as in \eqref{eq:ratio} will lead to the BFF decomposition.  

\subsection{Relevance of maximum circularity vs. BFF in different applications}
\label{sec:relevance}

These two decompositions serve different analytical goals. Choosing between them depends on what one wishes to learn from or do with the network.

\subsubsection{Structural Analysis and Strong Connectivity in Directed Networks}
\label{sec:structural_analysis_connectivity}

When the goal is to understand how circular flow is structured and distributed --- whether to detect feedback loops, identify likely recirculation pathways, or visualize weighted cycle patterns --- the BFF decomposition offers a principled and informative approach, but maximum circularity provides a structural extreme that will generally be unrepresentative. 

Maximum circularity (\ref{sec:maximum}) allocates flow to cycles that support the highest total circulation, often saturating some while minimizing or excluding others. This can exclude edges and nodes that participate in feasible cycles, leading to decompositions that are unrepresentative and that fail to capture full connectivity structure. While it offers a useful upper bound on circularity, maximum circularity selects cycles from a global, top-down perspective, overlooking the fact that real-world dynamics typically emerge from local interactions. Moreover, the maximum circularity solution is not necessarily unique, with different decompositions achieving the same total, so a unique solution then depending on edge selection or tie-breaking rules.

BFF by contrast, models a decentralized process in which nodes forward incoming flow across outgoing edges in proportion to the share of each edge in their total out-strength (section \ref{sec:BFF}). As a result, it spreads circulation across all feasible cycles in a capacity-respecting and distributed manner, more representative of how circulation patterns arise in real-world distributed or decentralized systems. 
Moreover, the BFF decomposition is unique.

BFF can also be seen as providing a capacity-sensitive measure of strong connectivity in directed networks. While strong connectivity --- defined as mutual reachability via directed paths --- is a key topological condition underpinning dynamics such as consensus \cite{Olfati-Saber_2007}, contagion \cite{Xu_2018}, and synchronization \cite{Chopra_2012, Lu_2007}, the standard definition is binary and insensitive to edge capacities \cite{Hamedmoghadam_2021}. Even a single low-capacity reverse edge can render a network strongly connected, overstating its ability to support actual circulation or feedback. BFF overcomes this limitation by preserving the full cycle topology of each SCC, while weighting it by edge capacities that can sustain circular flow.\footnote{Once the BFF circular weights $ c_{ij}$ are known, you can quantify feedback at two complementary levels: internal circularity of a SCC $\Omega$ as $ \frac{\sum_{ij\in \mathcal E_{\mathbf \Omega}}c_{ij}}{\sum_{ij\in \mathcal E_{\mathbf \Omega}}w_{ij}}$, or the contribution of each SCC to network wide circularity $ \frac{\sum_{ij\in \mathcal E_{\mathbf \Omega}}c_{ij}}{\sum_{ij\in \mathcal E_{\mathbf w}}w_{ij}}$, forming an additive decomposition of the total circularity and pinpointing which SCCs drive the network‑level ratio. Together these provide a high‑resolution map of capacity‑supported circulation, distinguishing how strong connectivity is within each strongly connected block, and how much it contributes to the system as a whole.} 
Future work could explore how BFF circularity correlates with real-world dynamics such as contagion persistence, resilience, and feedback-driven instability.

\subsubsection{Explicit flow allocation problems such as netting and routing}
\label{sec:explicit_flow_allocation}

Many applied tasks --- such as routing, or financial netting \cite{Srinivasan_1986,OKane_2017} --- permit or require 
the allocation of flow between circular and directional components. We note that CDFD is the space of all ``full netting'' solutions that leave no cycles (see section \ref{sec:compression}).\footnote{But not partial netting solutions, which only need to satisfy the constraints \eqref{eq:x_constrains}.} In general the directional part of any CDFD provides a valid routing of nett source flow to nett sinks without cycles and the CDFD solution space captures all feasible such routings. Even routing problems on acyclic structures can be reformulated as a circulation problems in the CDFD framework by first connecting consumers to suppliers
\footnote{From each node with nett demand to the auxiliary node and from the auxiliary node to each nett supplier, with weights equal to the flow that should enter or leave the network at the supplier and consumers respectively \cite{Ahuja_book}.}. Under this formulation the CDFD solution space (section \ref{sec:topology_decomposition}) inherently includes all valid flow allocations, underscoring its flexibility and broad applicability. 
The most appropriate decomposition within this space depends on the specific objectives and constraints of the task. Among the many possibilities, both maximum circularity and  BFF have broad relevance.

Maximum circularity is a natural choice when the goal is to reduce inefficiencies. In some cases, the objective is to minimize directional flow, for example, lowering gross exposures in financial netting or reducing transport distances in logistics. In others, the aim is to maximize circularity in order to minimize leakage and create the most ``closed” system possible, as in circular supply chains, recycling, or round-trip routing.

However, real-world systems often involve multiple and sometimes competing objectives. In financial netting, for instance, maximizing circularity reduces gross positions to the greatest possible extent, but may excessively concentrate exposures on a smaller set of counterparties, increasing counterparty concentration risk \cite{Yijie_thesis} (see discussion in Section~\ref{sec:compression} below).
Similarly, in routing problems, a purely cost-minimizing solution can funnel flow through the nodes or edges providing low cost paths, leading to localized overload, fragility, or premature depletion of critical resources \cite{Badiu_2021}. Thus where additional concerns such as diversification, resilience, or fairness are important, a maximum circularity solution may not be the best choice.

In such cases, and short of solving a more complex multi-objective optimization problem, the BFF algorithm offers a simple and practical compromise. It performs well in terms of netting or cost reduction, while distributing flows more broadly across the network. By avoiding severe concentrations and overuse of specific paths or nodes, it can potentially promote e.g. more balanced exposure, resource usage profile, or robust network performance – something that merits detailed investigation in future work.

\subsubsection{Example application:~Post-trade portfolio compression}
\label{sec:compression}

Consider a concrete example:~portfolio compression in over the counter (OTC) derivatives markets. This is a form of multilateral netting that allows counterparties to reduce their gross positions --- a major source of counterparty risk and driver of regulatory capital requirements --- without altering nett positions \cite{DErrico_2021}. Compression is typically facilitated by a post-trade service provider, who is trusted with access to the full set of contract relationships and uses this knowledge to propose a compression that must be accepted by all participants in order to be implemented. Compression is a well-established market practice that has had a significant impact on these important markets \cite{Schrimpf_2015}. For instance, leading service provider OSTTRA reports having eliminated more than \$2,400 trillion in notional gross exposure through its TriReduce service.\footnote{See OSTTRA website: https://osttra.com/services/optimisation/portfolio-compression/}

Within this context, the directional component of any CDFD corresponds to a valid compression. Any of these valid compressions will not only reduces excess gross positions for participating nodes but also eliminate strong connectivity (entirely captured by the circular part) which may be a source of contagion risk \cite{Yijie_thesis, DErrico_2017}. 

Selecting the maximum circularity will yield the greatest possible reduction in gross notional. However, because maximum circularity concentrates flow in a subset of minimum cost paths, this can translate into the netting of contracts with some counterparties and not others, concentrating portfolios in fewer counterparties resulting in an undesirable pattern of exposure, and the asymmetric distribution of netting benefits among participants \cite{Yijie_thesis}. By contrast, the BFF algorithm explicitly aims to preserve the relative pattern of exposures and ensures that compression gains are shared more evenly. By utilizing all available cycles, it guarantees that every participant who can benefit from compression does so.

Moreover, the local version of BFF in principle allows for privacy-preserving compression mechanisms. Unlike traditional services that require a trusted intermediary with full visibility into participants’ portfolios, a local implementation of BFF could enable compression without requiring full portfolio disclosure. In principle, it also opens the door to decentralized compression mechanisms that do not rely on a central trusted third party --- particularly relevant in contexts such as decentralized finance (DeFi) or other settings where trustless coordination and data confidentiality are paramount. While privacy-preserving and decentralized netting have been topics of academics and practitioner interest, practical solutions have remained elusive \cite{Cao_2020, Bottazzi_2024,Patsinaridis_2018}, making this a promising direction for further research and development.

A full comparison between maximal circularity and BFF algorithm in the financial context is presented in \cite{Yijie_thesis}.

\section{Numerical results}
\label{sec:numerics}

In this section we will show how the proposed measures of circularity behave for large synthetic and real networks. We have already seen the example in figure~\ref{fig:first_eg}, where the maximum circularity is 6/8 and the BFF circularity is 3/8 (achieved at $\alpha = 1/2$). We now examine a slightly larger network to understand the possible shape of the entire decomposition space.

\subsection{Example decomposition space}
\label{sec:ex_decomposition_space}

\begin{figure}[ht]
     \centering
     \begin{subfigure}[b]{0.365\textwidth}
         \centering
        \includegraphics[width=\textwidth]{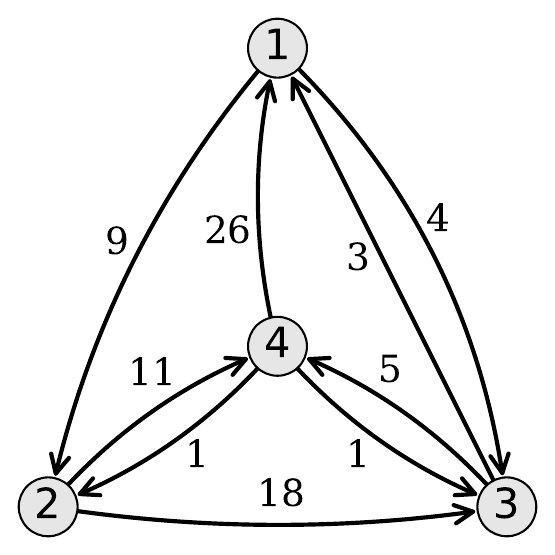}
     \end{subfigure}
     \hfill
     \begin{subfigure}[b]{0.585\textwidth}
         \centering
         \includegraphics[width=\textwidth]{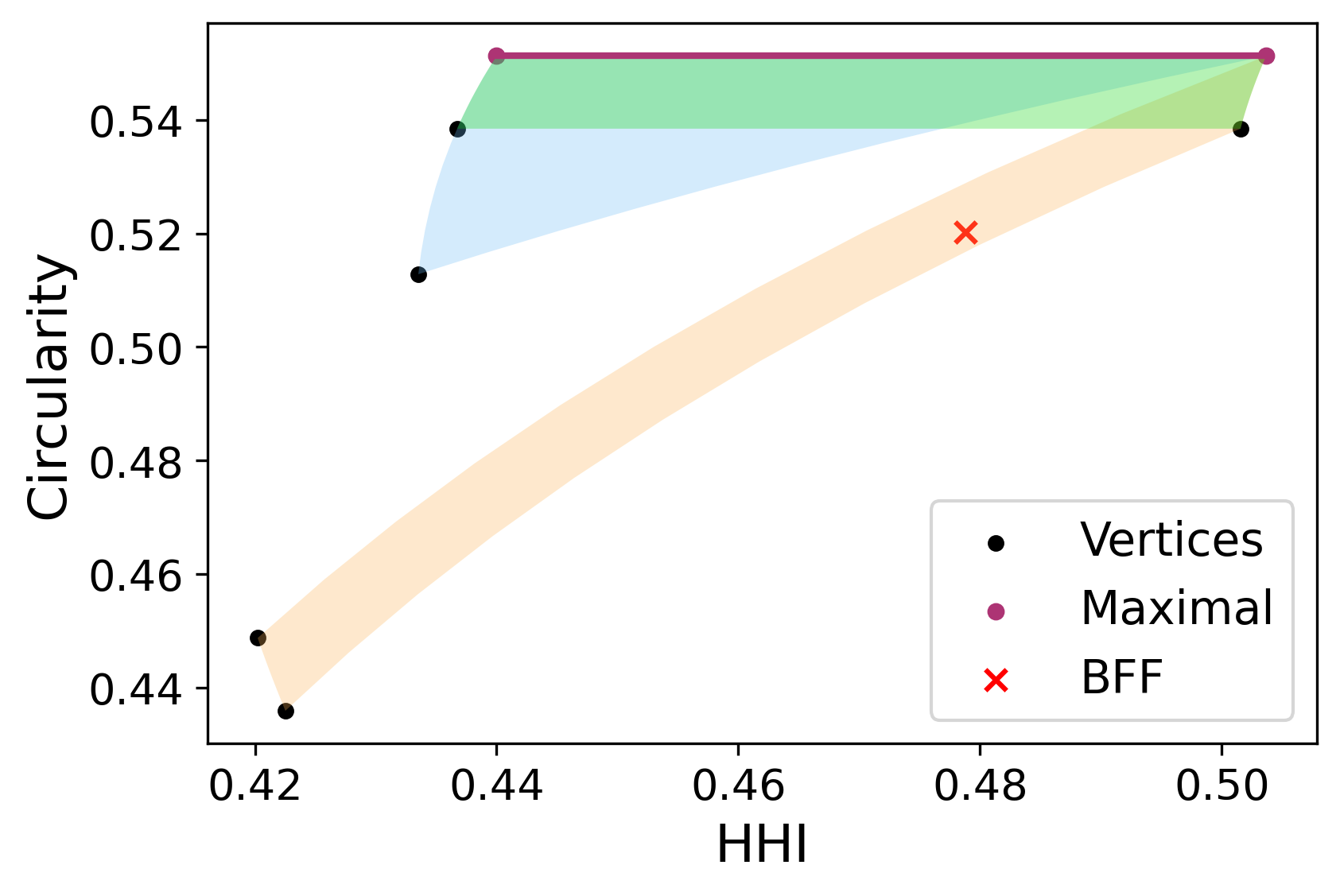}
     \end{subfigure}
        \caption{ Decomposition space (on the right) of the network depicted on the left. 
        The shaded regions in green and orange correspond to polytopes with 4 vertices whereas the blue one has 3. Maximal circularity is achieved by two vertices, depicted in purple, and thus also by the segment joining them. }    
        \label{fig:interesting_D}
\end{figure}

The natural embedding of the decomposition space is high dimensional, as many dimensions as edges, so it is necessary to project it down to visualise it. We choose as coordinates the circularity \eqref{eq:circularity_ratio} and the Herfindahl–Hirschman Index (HHI) in weight which is a concentration measure widely used in financial markets \cite{HHI_index}. For a network $\mathbf w$ this index is defined by 
\[\textnormal{HHI} = \sum_{i\in \mathcal N_{\mathbf w}} \left ( \frac{w_i ^\outd }{w_{\textnormal{total}}}\right )^2,\]
where $w_{\textnormal{total}} = \sum _{ij\in \mathcal E_ {\mathbf w} }w_{ij} $. The HHI ranges from $1/|\mathcal{N}_{\mathbf w}|$ to 1, where low values indicate high diversification. 

In figure~\ref{fig:interesting_D} a minimalistic example of a network that has an interesting decomposition space is depicted. The shaded regions representing polytopes are slightly deformed due to the  non-linearity of HHI, but recall that they are convex in their ambient space. 
Maximal circularity is achieved by two vertices (depicted in purple) and thus also by the segment joining them. The circular parts corresponding to these vertices differ only in the usage of weight 3 in the cycle 1, 2, 3 instead of the cycle 1, 2, 4. These cycles have the same length, thus explaining the equality. Minimum circularity is achieved in a single vertex, giving us a circularity range of $[0.44,0.55]$. We propose that the BFF decomposition which is in the interior of the orange polytope and with a circularity $0.52$ is a good representative of the space. 

\subsection{Circularity in random networks}
\label{sec:circularity_random}

\begin{figure}[ht]
    \centering
\includegraphics[width=0.88\textwidth]{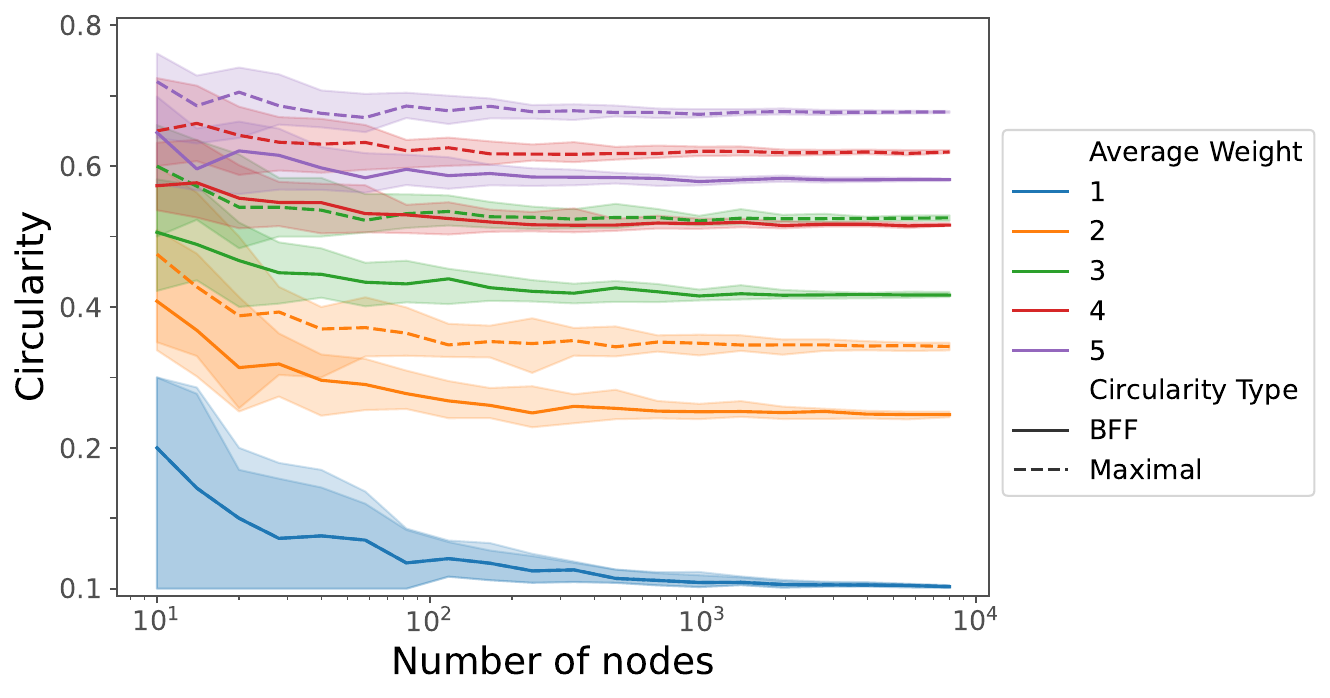}
\caption{Maximal and BFF circularity for random networks of increasing size for various average weights. The middle lines represent the medians and the shaded region the values between the first and third quartiles over 50 simulations. } 
    \label{fig:circularity_vs_num_nodes}
\end{figure}

Our random ensemble is a slight variation of the famous Erdős–Rényi ensemble, to give weighted directed networks. Given the number of nodes $n$ and edges $m$, we uniformly pick $m$ directed edges $ij$ with $i\not = j$ and  weight 1. A given edge $ij$ can be generated more than once, in which case we amalgamate them by adding their weights.

First we study how the circularity changes on increasing the size of the network, while keeping its average weight\footnote{That is the average in/out-weight or equivalently the total weight divided by number of nodes.} $k$ fixed, i.e.~we take $m = k n $. 
We find in figure~\ref{fig:circularity_vs_num_nodes} that besides a somewhat irregular start for small networks, the circularity values remain quite constant. 

When the average weight is $k = 1$, we see that both circularities tend to 0. This is not surprising as for $k<1$ the expected number of edges contained in cycles is finite, making the expected circularity tend to zero as it is normalized, see appendix~\ref{app:Random_networks}. 
For the other average weights we see a significant difference (around 0.1) between the maximal circularity and that attained by the BFF algorithm. Thus, we can infer that these networks have rich decomposition spaces.  

\begin{figure}[ht]
    \centering
\includegraphics[width=0.99\textwidth]{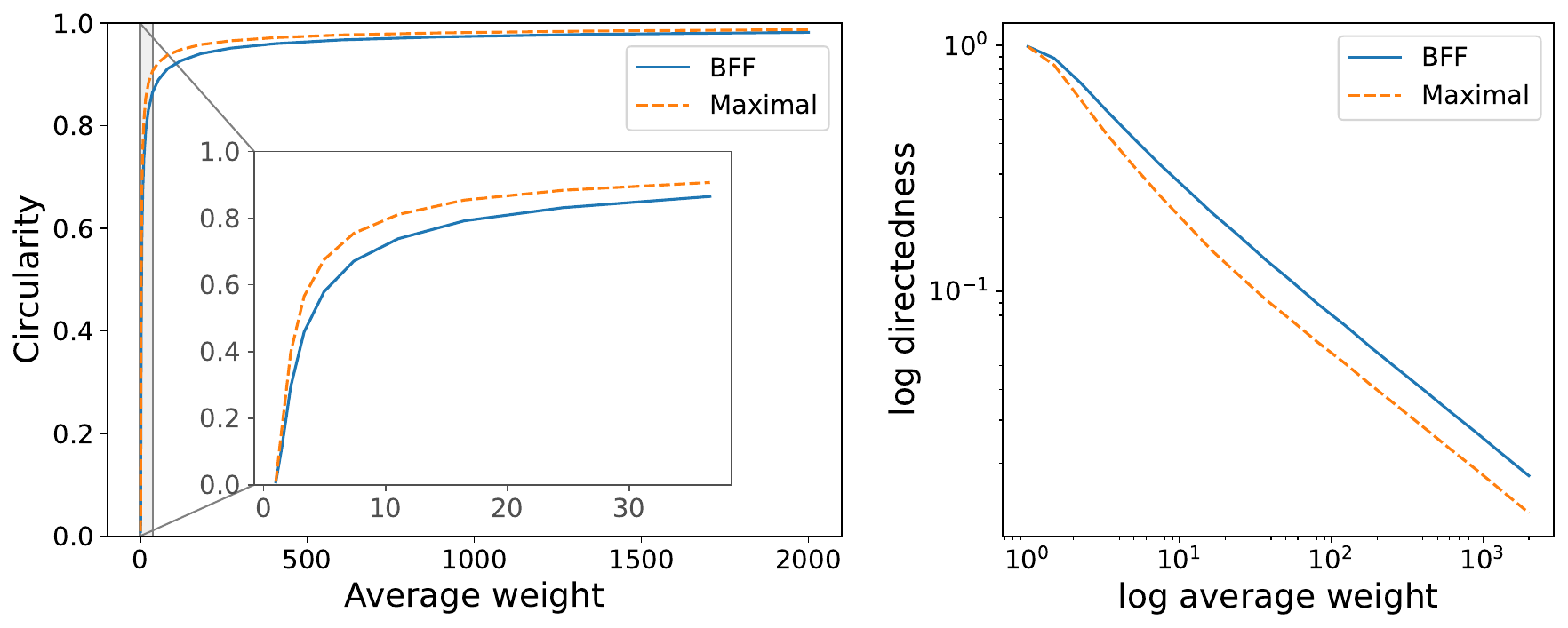}
\caption{On the left, the mean circularity for random networks with 2000 nodes from sparse to dense regime. An insert has been added to highlight the sparse regime.  On the right, the log-log plot of the mean directedness is depicted.}
    \label{fig:circularity_vs_average_degree}
\end{figure}

Now we plot circularity versus the average weight, see figure~\ref{fig:circularity_vs_average_degree}, for a fixed network size as we have seen that the number of nodes has little effect. For both maximal and BFF circularity we see a distinctive concave shape going from 0 at average weight 1 and tending to 1 at infinity. Both of these values can be deduced theoretically, see appendix~\ref{app:Random_networks}. These plots can be fit by a power law behaviour of the directionality $b k^{-a}$ or equivalently $1-bk^{-a}$ for circularity. We check this relation by making a log-log plot of directedness in figure~\ref{fig:circularity_vs_average_degree}(b),  where we indeed see approximate straight lines.

Note that for low average weight our random model behaves quite similarly to the unweighted directed ER model as most of the edges will have weight 1, see appendix~\ref{app:Random_networks}. 
This is no longer true for high densities. For instance for $k=|\mathcal N|$ in the unweighted ER one always obtains a complete graph with circularity 1 so that circularity cannot obey a power law for dense networks.

\subsection{Empirical example:~quantifying circulation of a digital community currency}
\label{sec:digital_currency}

Sarafu is a digital community currency active in Kenya. As its tokens are not legal tender, circulation offers a meaningful indicator of the currency’s practical success \cite{Mattsson_2023}. We analyze a weighted directed network of transactions (in Sarafu units) between users from January 2020 to June 2021, encompassing over 400,000 transactions totaling 293.7 million Sarafu among approximately 40,000 users. For full data documentation, see \cite{Mattsson_2022_data}.

Following \cite{Mattsson_2023}, we time-aggregate the transactions and apply the Infomap algorithm \cite{Rosvall_2008} to identify modular structure. The five largest modules --- corresponding roughly to distinct geographic regions --- capture 99.7\% of the total transaction volume. In their analysis, \cite{Mattsson_2023} studied circularity by counting cycles of lengths 2 to 5 within modules, and found significantly elevated cycle counts (especially for lengths 2 and 3) relative to unweighted null models, suggesting a high degree of topological circularity within the regions. We revisit this question using our CDFD-based approach, which accounts for both edge weights and cycles of arbitrary length. While \cite{Mattsson_2023} highlight the abundance of short cycles, our method --- being global, weight-sensitive, and normalized --- quantifies the extent of actual monetary circulation in the network. In table~\ref{tab:sarafu_modules} we observe high levels of circularity within individual modules, particularly among those of lower rank, which correspond to more densely interconnected communities. Furthermore, we are able to quantify the overall circularity of the network, obtaining a value of 0.880 under the BFF decomposition and 0.899 under the maximum circularity decomposition.

\begin{table}[ht]
\centering
\begin{tabular}{clcc}
\toprule
\textbf{Module} & \textbf{Region} & \textbf{BFF} & \textbf{Maximal} \\
\midrule
 1 & Mukuru Nairobi & 0.926 & 0.945 \\
 2 & Kinango Kwale & 0.801 & 0.817 \\
 3 & Misc Nairobi, Kilifi \& Nyanza & 0.746 & 0.772 \\
 4 & Kisauni Mombasa & 0.646 & 0.666 \\
 5 & Turkana & 0.684 & 0.687 \\
\bottomrule
\end{tabular}
\caption{Circularity of the five largest modules identified by Infomap. The table also indicates the approximate geographic region each module represents. Note that Module 3 spans localities from three distinct regions. } 
\label{tab:sarafu_modules}
\end{table}

We find little difference between the BFF and maximal circularity values across all modules, in contrast to the variations reported in previous sections for random networks. This pattern persists even when controlling for the number of nodes and edges. To explain this phenomenon, we leverage our methodology, which not only computes a circularity value but also provides a circular flow decomposition. Analyzing these decompositions reveals that the Sarafu networks exhibit a remarkably simple circularity structure, dominated by short cycles, particularly those of length two. This is consistent with the findings in \cite{Mattsson_2023}
and may reflect the prevalence of bilateral economic relationships between individuals. In contrast, our sparse network ensembles are expected to contain relatively few short cycles (see Appendix~\ref{app:Random_networks}). Note that cycles of length two are fully captured by all circularity measures, as they are entirely included within the circular component of any decomposition. Therefore, in networks where circularity is predominantly driven by such interactions, as is the case with this dataset, we should expect similar circularity values regardless of the measure used.

Finally, it is also valuable to examine the circulation of money over shorter time intervals. We defer this analysis to Section~\ref{sec:numerics_comparison}, where we present monthly snapshots of the Sarafu dataset in figure~\ref{fig:correlation}.

\subsection{Comparison with other measures}
\label{sec:numerics_comparison}

In this section, we present a numerical comparison of CDFD based circularity measures against normalized circularity metrics from the literature. Specifically, we benchmark against:~trophic incoherence \cite{how_directed, Kichikawa2019} which is within the HHD framework and is increasingly being used as a measure of circularity and directedness; LM circularity, the complement to the directedness measure introduced by Luo et al.~\cite{Luo_2011}; and the FCI, originally proposed in ecology \cite{Finn_1976,Braun_2021} and widely used in other domains (see Section~\ref{sec:motivation_comparison} for a detailed review). For this comparison we employ the two specific decompositions highlighted in Section~\ref{sec:choosing_rep}: maximum circularity and BFF.\footnote{Note we do not consider the Ulanowicz-decomposition here \cite{Ulanowicz_1983} (which also conforms to a CDFD) as it relies on explicit cycle enumeration (which scales exponentially with network size), depends on some arbitrary choices, and lacks clear conceptual motivation as a specific choice of CDFD.}

For the comparison we employ a diverse pool of empirical datasets. These include:~food webs (where edge-weights represent the flow of biomass between species) \cite{Lin_2024}; transaction networks from the Sarafu community currency in Kenya (where weights represent the value of payments) \cite{Mattsson_2022_data}; national input–output networks (capturing inter-sectoral flows of value added for different economies) \cite{how_directed}; and government web-link networks between U.S. state-level agencies (used to study the structure of government organization based on online presence) \cite{Kosack_2018}.

\begin{figure}[ht]
    \centering
\includegraphics[width=0.9\textwidth]{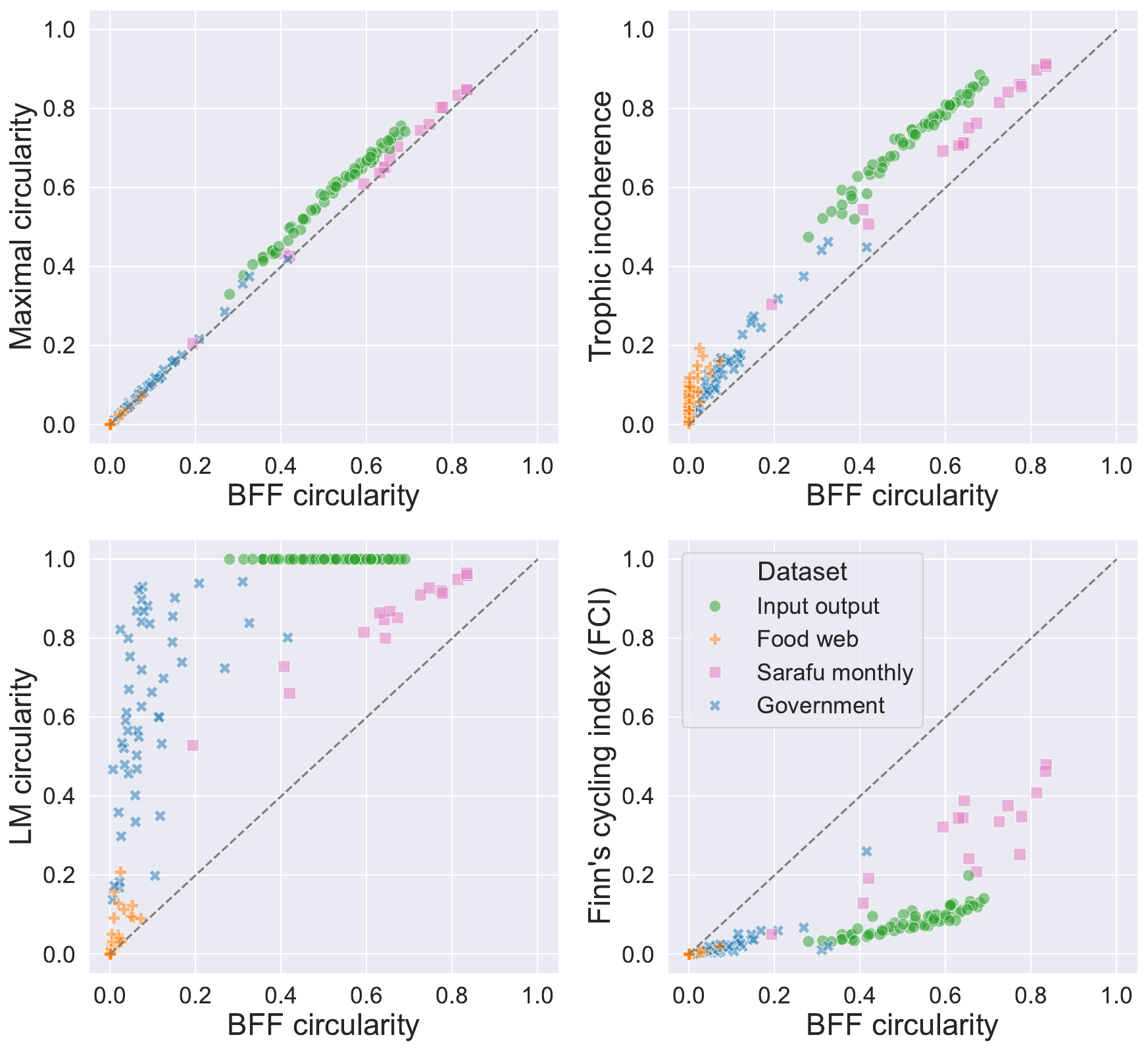}
\caption{Correlation plots between different notions of circularity. }
    \label{fig:correlation}
\end{figure}

By applying the measures across these datasets, we are able to assess their behavior across systems with different underlying structures. This analysis is plotted in figure~\ref{fig:correlation}.

We see max circularity and BFF move closely together across all datasets. Max circularity is always greater than or equal to BFF, as expected, but the gap widens for the dense input–output networks. In these networks, the rich connectivity allows max circularity to route directed flow almost bipartitely, while BFF distributes weight more evenly across the full connectivity structure, better reflecting the interdependence between different industries.

Because the two measures are so tightly correlated, in the comparisons with other metrics we show only BFF; it would not add anything or change the conclusions to also show comparison with max circularity.

BFF and trophic incoherence are strongly correlated, as cycles undermine the possibility of strict hierarchical ordering. This somewhat justifies its use as a proxy for circularity  ---  or, conversely, for directedness. However, the dispersion visible in the food-web cloud illustrates a key limitation. Many of these food webs are entirely acyclic (with zero BFF), yet display substantial variation in incoherence. This variation arises not from circular flow, but from feed-forward structures. Strictly speaking, incoherence measures the extent to which a network departs from a perfectly layered hierarchy, making it an imperfect and potentially misleading proxy for circularity.

For LM circularity, the relationship with BFF is much weaker. In the production network data, every sector connects to every other sector, so the graphs are strongly connected and LM circularity saturates at 1, offering no discriminatory power. Meanwhile, BFF varies significantly (from approximately 0.2 to 0.7), indicating meaningful differences in circularity across countries. In contrast, government-web graphs show the opposite pattern: LM circularity suggests government structures vary a lot (approximately 0.1 to 0.9), however BFF clusters around 0.1. This shows that even in government networks that exhibit high topological circularity, references are mostly in one direction. These cases illustrate the significant limitations of LM circularity. 

Finally, as theory predicts, FCI and BFF coincide at the two extremes: for acyclic networks, FCI=BFF=0; for pure circulations, FCI is undefined but tends to 1 by continuity, matching BFF. Between these extremes, the relationship is nonlinear: for weak feedback, FCI grows more slowly than BFF (because inflow ``leakage'' dominates the Neumann series expansion), while for stronger feedback (above approximately 0.5 BFF), FCI steepens and converges toward BFF as the network approaches a circulation. FCI may be best suited for applications where a residence-time under Markov dynamics interpretation is meaningful and relevant.

\section{Conclusions}
\label{sec:conclusions}

We introduced the Circular Directional Flow Decomposition (CDFD), a general framework for partitioning any weighted directed network into two components:~a directional (acyclic) flow and a circular (divergence-free) flow. This yields a natural, dimensionless circularity measure, ranging from 0 (purely directional) to 1 (fully circular), that quantifies the extent of circulation in the system. While the decomposition is not unique, we show that the space of valid decompositions has favourable geometric and topological properties, and its connectedness reflects the system’s intrinsic capacity for circulation.

Through numerical simulations, we demonstrate that decomposition spaces can span wide circularity ranges, highlighting the importance of selecting appropriate representatives. We propose two polynomial-time benchmark decompositions:~maximum circularity, which globally minimizes directional flow, and Balanced Flow Forwarding (BFF), a unique, distributed solution that proportionally reflects local structure. While both are valid, they serve different purposes:~because BFF arises from local proportional forwarding, it is suited for quantifying and analyzing the structure of circulation in a broad range of systems. In contrast, maximum circularity may be preferable in top-down applications prioritizing efficiency, such as financial netting or closed-loop logistics. Yet even in top-down contexts, BFF may be advantageous where local fairness or diversification matters (e.g.~avoiding counterparty concentration or promoting balanced resource use).

The CDFD framework is broadly applicable across systems where circulation and directionality matter --- from nutrient cycles and metabolic pathways to financial infrastructure and supply chains. Practical uses include routing, netting, and resilience analysis, with localized BFF variants especially suited to decentralized settings. Directions for future methodological work include node-level circularity metrics, detection of internally circular communities, sampling of the decomposition space,  flow-rate-based variants, and extensions to handle negative flows. It would be also worth considering other possible canonical choices of decomposition, such as taking the $L^2$-closest circular flow, and developing computationally tractable methods for approximating the minimum circularity CDFD. We hope this framework supports both theoretical insight and practical innovation across domains.

\vspace{10pt}

\textbf{Data accessibility:} The dataset introduced in section~\ref{sec:digital_currency} is available via the UK Data Service (UKDS) under their End User License, 
which stipulates suitable data-privacy protections. The dataset is available for download from the UKDS ReShare 
repository (\url{https://reshare.ukdataservice.ac.uk/855142/}) to users registered with the UKDS. Data cleaning has been done following \cite{Mattsson_2023}. Additional datasets used in Section \ref{sec:numerics_comparison} were obtained from the sources cited therein, all of which have made the data publicly available for download (food webs \cite{Lin_2024}; input–output networks \cite{how_directed}; government web-link networks \cite{Kosack_2018}). The code required to reproduce each figure and analysis is available for download at this Github repository \texttt{https://github.com/marchd1997/CDFD}. A Python and Sage toolbox is also available to facilitate the implementation of the algorithms introduced in this work.

\AtNextBibliography{\small}
\printbibliography

\appendix

\section{Minimum cost flow equivalence}
\label{app:minimu_cost_equivalence}

Following the notation in \cite{Ahuja_book}, a potential on a network is an assignment of real values on its nodes. We say that a potential $\phi$ is a topological ordering if following any edges in the network increases the potential. That is, $\phi (i)<\phi(j)$ for each edge $ij\in \mathcal E_{\mathbf w}$. Not all networks have topological orderings; in fact, it is not hard to show that that acyclic networks are precisely the ones that have a topological ordering \cite[p.77]{Ahuja_book}.

Potentials help us find optimal solutions to \eqref{eq:x_constrains} thanks to the following result \cite[Theorem 9.4]{Ahuja_book}. 
\begin{lem}
A network $\mathbf x ^*$ satisfying \eqref{eq:x_constrains} is a minimum of \eqref{eq:min_x} iff there exists a potential $\phi$ in $\mathbf w$ such that
\begin{subequations}
\begin{align}
    \kappa_{ij}^\phi>0 \hspace{0.5cm}&\Longrightarrow \hspace{0.5cm}x^*_{ij}=0,\label{eq:lem_kappa_ij^phi_1}\\
    \kappa_{ij}^\phi<0 \hspace{0.5cm}&\Longrightarrow\hspace{0.5cm} x^*_{ij}=w_{ij},\label{eq:lem_kappa_ij^phi_2}\\
    0< x^*_{ij}<w_{ij}\hspace{0.5cm} &\Longrightarrow\hspace{0.5cm} \kappa_{ij}^\phi=0,\label{eq:lem_kappa_ij^phi_3}
\end{align}
\label{eq:lem_kappa_ij^phi}
\end{subequations}
where $\kappa_{ij}^\phi = \kappa_{ij}-(\phi(j)-\phi(i))$.
\end{lem}

We can now prove theorem~\ref{thm:CDFD_eqi_min_flow}.

\begin{proof}[Proof of theorem~\ref{thm:CDFD_eqi_min_flow}]
Assume an optimal solution $\mathbf x ^*$ of \eqref{eq:min_x} contains a cycle. Then, removing from $\mathbf x ^*$ a small amount of flow through it will  preserve the restrictions \eqref{eq:x_constrains} while reducing the total cost, which contradicts $\mathbf x ^*$'s optimally. Thus,  $\mathbf x ^*$ is acyclic and as it satisfies \eqref{eq:x_constrains}, $\mathbf w-\mathbf x ^*$ is balanced. So $\mathbf x ^*$ is a directional part of $\mathbf w$.  

For the other direction consider a directional part $\mathbf d$ of $\mathbf w$, and we show that it satisfies \eqref{eq:lem_kappa_ij^phi} for appropriate choice of costs $\vect \kappa> 0$ and potential $\phi$. As $\mathbf d$ is  acyclic, we  let $\phi$ be a topological ordering of its nodes. For edges  $ij\in \mathcal E_{\mathbf d}$, we let $\kappa_{ij}=\phi(j)-\phi (i)>0$ so that  $\kappa_{ij}^\phi=0$. Then \eqref{eq:lem_kappa_ij^phi} is satisfied regardless of the value of $d_{ij}$. For edges in $\mathbf w$ but not in $\mathbf d$ we have $d_{ij} = 0$. Then \eqref{eq:lem_kappa_ij^phi} reduces to \eqref{eq:lem_kappa_ij^phi_2}, and to complete the proof it is enough to guarantee that the condition on the left hand side is never satisfied, i.e. $\kappa_{ij}^\phi\geq 0$. This can be done by simply choosing $\kappa_{ij}>0$ large enough.

\end{proof}

\section{Polytope complexes and their topology}
\label{app:background}

\subsection{Polytope complexes} 
\label{app:background_polytope}

Given points $\mathbf x_1, \dots , \mathbf x_k\in \mathbb R^m$, we say that $\sum_{i=1}^k \lambda_i \mathbf x_i$ is a convex combination of them if $\sum_{i=1}^k \lambda_ i = 1$ and $\lambda_i\geq 0$ for all $i$. Given $M \subset \mathbb R ^m$ its convex hull, $\conv  M$, is the set of all convex combinations of points in $ M$. We say that $ M$ is convex if $\conv  M =  M$.  

We will be interested in sets that are the convex hull of finitely many points on $\mathbb R^m$,  which are known as convex polytopes or simply  as \emph{polytopes}. It turns out that polytopes can also be characterized as bounded sets given by a finite number of linear equalities and non-strict inequalities. 
The dimension of a polytope $ P$, $\dim P$, is the dimension of the smallest affine space that contains it.  
One of the most recognisable features of polytopes are their faces, which are given by
the intersections of the polytopes with  hyperplanes that do not divide them into two parts.
Equivalently, and more relevant for our work,  $F\subsetneqq P$  is a proper \emph{face} of a polytope $ P$ if it minimises some linear functional, that is, 
\begin{equation}
    F = \operatorname*{arg\,min}_{\mathbf x \in  P} \sum_{i=0}^m \kappa_i x_i
    \label{eq:arg_min_def_F}
\end{equation}
for some $\vect \kappa \in \mathbb R^m$. It will be convenient to consider $\emptyset$ and $ P$ as  improper \emph{faces}. It is clear from this definition that faces are themselves polytopes. We note that $ P$ can also  be seen as the minimizer of the zero functional, and it can be shown that it is the only face of $P$ with dimension $\dim P$.
 The zero-dimensional faces, i.e. faces formed by a single point, are called \emph{vertices}, and we denote the set of them by $\ver P$. It can be shown that all faces can be expressed as the convex hull of the  subset of vertices they contain, and in particular $\conv (\ver P)= P$. 
 The collection of all faces of a polytope has a lot of structure, which can be summarised by the fact that it forms a polytope complex. 

A \emph{polytope complex} is a finite family $\m K$ of polytopes, called \emph{cells} of $\m K$, such that
\begin{enumerate}[(i)]
    \item each face of a member of $\m K$ is itself a member of $\m K$;
    \item the intersection of two members of $\m K$ is a face of both members. 
\end{enumerate}
Cells are partially ordered by inclusion and we denote by $\Max \m K$, the set of maximal cells under this relation.  We also denote by $\ver \m K$ the union of the vertices of all the cells in $\m K$. Polytope complexes are a generalization of the better-known \emph{geometrical simplicial complexes}, which are formed by \emph{simplices}, i.e. polytopes with dimension one less than the number of vertices, and satisfy the same relations (i), (ii). 

The essential information of a simplicial complex can be retained by simply recording each of its simplices by the set of vertices they contain. This gives us an \emph{abstract simplicial complex}, that is, a collection $\m A$ of subsets of a finite set $V$ such that if $S\in \m A$ then any subset of $S$ is also in $\m A$. We refer to $V$ as vertices and denote them by $\ver \m A$, even if the complex is not geometrical. In fact, we can always find a geometrical realisation of $\mathcal A$, which we denote\footnote{We define this up to homeomorphism, as there is more than one possible geometrical realisation but all of them are homeomorphic.} by $\geo \mathcal A$, which is a geometrical simplicial complex whose abstract representation is isomorphic to $\mathcal A$. We say that $f: \m A_1 \rightarrow \m A_2$ is a \emph{simplicial isomorphism} if $f:\ver \m A_1\rightarrow\ver \m A_2$ is a bijection and $S\in \m A_1$ iff $f(S)\in \m A_2$. If such a map $f$ exists the simplicial complexes are isomorphic and they have the same geometrical realizations.  

For a more in-depth look at the concepts defined in this section see \cite{convex_polytopes} and \cite{gallier2008notes}. 

\subsection{Topology}

When we talk about the topology of a polytope complex $\m K$ we mean the topology of the space formed by the union of all the cells, i.e.~the space $\bigcup \m K:=\cup_{F\in \m K}F$. Note that it is enough to take the maximal cells as $\bigcup \m K = \bigcup \Max \m K$. Similarly the topology of an abstract simplicial complex is given by its geometric realization. One may hope to study this space up to homeomorphism, which we denote by $\cong$, but in practice one usually uses the weaker notion of homotopy equivalence. Intuitively, two spaces $X$ and $Y$ are \emph{homotopy equivalent}, and we denote it by $X\simeq Y$, if they can be transformed into one another by bending, shrinking and expanding operations, see \cite{Hatcher} for a formal definition. We say that a space is \emph{contractible} when it is homotopically equivalent to the space with a single point $\{*\}$, so it has the simplest possible topology under this relation. Contractibility implies connectivity and path connectivity, but also the absence of holes in any dimension. For instance, convex sets are contractible. 

In this work the only tool we will use to show homotopy equivalence is the following version of the nerve lemma. Given a finite collection of sets $\mathcal X$ we define its \emph{nerve} as the abstract simplicial complex
\[\ner \mathcal X = \left \{\mathcal Y \subset \mathcal X \; | \;  \bigcap \mathcal Y \not = \emptyset  \right \},\]
where we use the notation $\bigcap \mathcal Y : = \bigcap_{Y\in \m Y} Y$. The following  deep connection between this abstract simplicial complex and the original space will be the key for the topological study of our spaces, see \cite{Nerve_lemma}.

\begin{lem}[Nerve Lemma]
Let $\mathcal X$ be a collection of convex compact sets. Then, the nerve of $\mathcal X$ is homotopically equivalent to the union of sets in $\mathcal X$. That is,
\[\ner \m X \simeq \bigcup \m X.\]
\label{lem:nerve}
\end{lem}

\section{Non-convex decomposition space}
\label{app:non-convex}

Let us consider the network depicted on the left of figure~\ref{fig:eg_non_convex}. By recursively 
picking one of its cycles and removing the maximum allowed capacity from it until no more cycles are left, we find ``extreme'' decompositions of it, i.e.~the vertices of $D$. By doing this in different orders, we find three directional parts of the network that are depicted on the right of figure~\ref{fig:eg_non_convex}. The convex combination of some of these directional parts give us other decompositions, which are represented by thick solid lines. Note that there is no line connecting the bottom ones, as convex combinations of them would contain a cycle between nodes 2 and 3. In particular the decomposition space is not convex and consists of two segments joined in a vertex as sketched in figure~\ref{fig:eg_non_convex}.

\begin{figure}[ht]
    \centering
\includegraphics[width=0.8\textwidth]{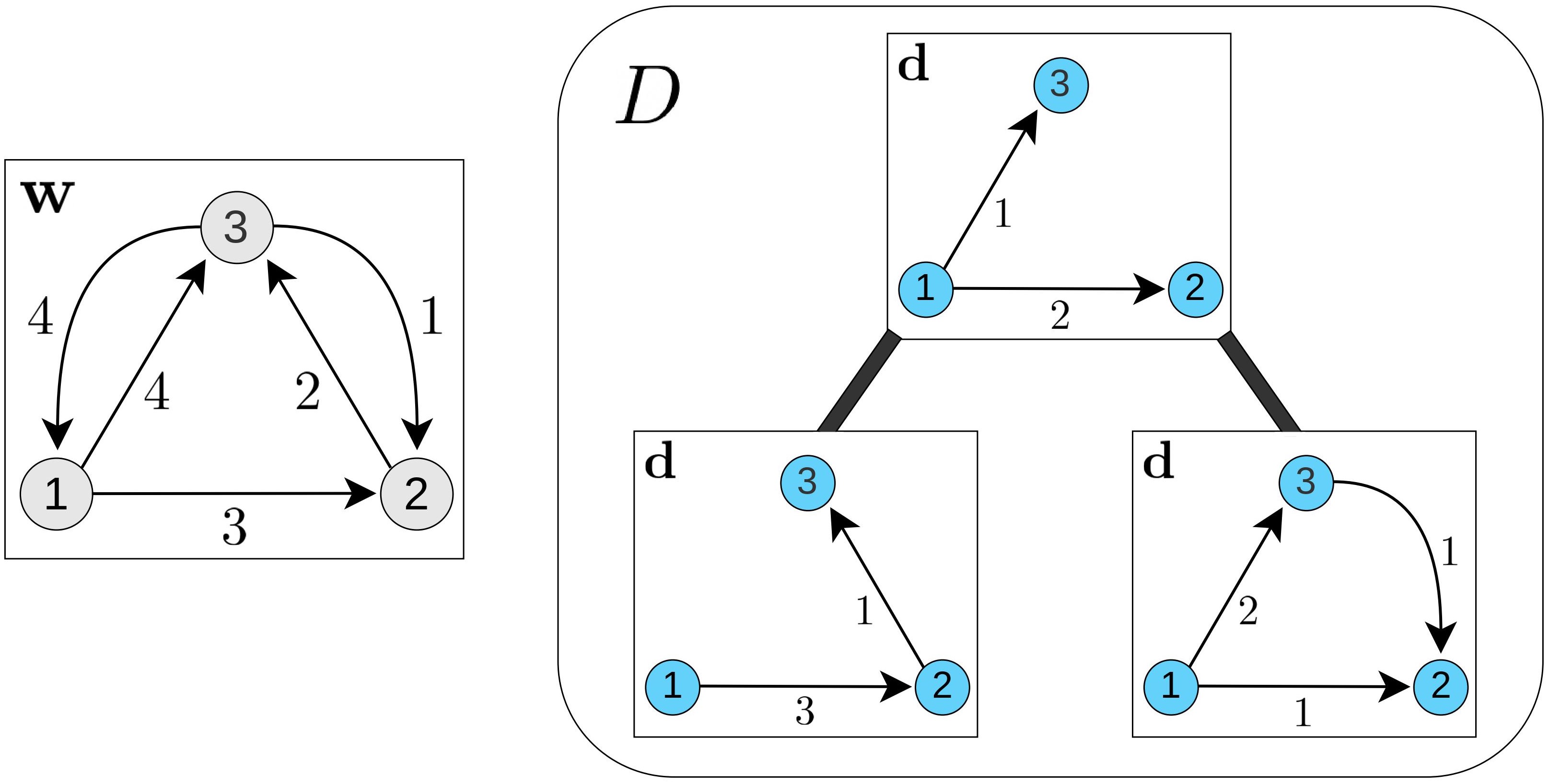}
\caption{ Depiction of the decomposition space $D$ of the network $\mathbf w$ shown on the left. We represent the networks corresponding to vertices of $D$ inside the squares and the thick black lines represent segments of $D$ connecting the vertices.}
\label{fig:eg_non_convex}
\end{figure}

\section{Decomposition space is a contractible polytope complex }
\label{app:decomposition_space}

This section relies heavily on the background presented in appendix~\ref{app:background}. 
First of all note that the linear constraints \eqref{eq:x_constrains} define a polytope $P$, as the set satisfying them is clearly bounded.  We will denote the argument minimum sets by
\[\amin (\vect \kappa) : =   \operatorname*{arg\,min}_{\mathbf x \in P} \sum_{ij\in \mathcal E_{\mathbf w}}\kappa_{ij}x_{ij},\]
which are faces of $P$ as defined in \eqref{eq:arg_min_def_F}. 
Then theorem~\ref{thm:CDFD_eqi_min_flow} states that $D$ is formed by the union of faces of $P$ that minimize some positive cost. We would like to show that the collection of these faces, which we denote by 
 \[\mathcal D := \{\amin (\vect \kappa) \;|\; \vect \kappa> 0\}\cup \{\emptyset\},\]
 forms a polytope complex. First note that polytopes, and in particular $P$, have finitely many faces, so the definition of $\mathcal D$ can actually be taken over finitely many costs. The rest is shown in the following result. 
 \begin{lem}
 We have 
 \begin{equation}
     \mathcal D = \{F \subset D \;|\; F \textnormal{ is a face of } P\}
     \label{eq:F_polytope_complex}
 \end{equation}
 and thus it is a polytope complex.
     \label{lem:F_polytope_complex}
 \end{lem}
 \begin{proof}
     We know that the set of  faces of $P $ forms a polytope complex, so it is clear that the  right hand side of \eqref{eq:F_polytope_complex} is also a polytope complex.  We have just argued above the lemma that $\mathcal D $ consists of faces of $P$ contained in $D$. 
     So all that is left to show is the other inclusion in \eqref{eq:F_polytope_complex}.

Letting $F \subset D$ be a proper face of $P$, we need to find $\vect \kappa> 0$ such that $\amin (\vect \kappa) = F$. 
     By definition of face, we know there exists $\vect \kappa_1$, not necessarily with all positive entries, such that $\amin (\vect \kappa_1) = F$. We also know that $\m D$ covers $D$  with finitely many faces of $P$ so there exists $\vect \kappa_2> 0$ such that  $\dim ( \amin(\vect \kappa_2) \cap F)=\dim F $. But then $\amin(\vect \kappa_2) \cap F$ is a face of $F$ having the maximal dimension, so  actually it is $F$. In particular, $F\subset \amin(\vect \kappa_2)$.
    Consider $\vect \kappa=\varepsilon \vect \kappa_1 + (1-\varepsilon) \vect \kappa_2$ with $0<\varepsilon<1$, small enough such that $\vect \kappa> 0$. This is a cost with the desired property as we have
    \[\amin (\vect \kappa)  = \amin (\vect \kappa_1)\cap \amin ( \vect \kappa_2) = F \]
    where in the first equality we have used linearity and that the intersection is non-empty.  

    The case of non-proper faces is trivial, so we have equality in  \eqref{eq:F_polytope_complex} as desired. 
 \end{proof}

We now reduce the topological study of polytope complexes to the more familiar simplicial complexes. We do this for $\m D$, but  all arguments can be applied to an arbitrary polytope complex\footnote{This can also be done using a triangulation of the polytope complex. However, this construction requires some arbitrary choices which would not be adequate for our subsequent proofs. }. Consider the abstract simplicial complex attained by replacing each cell of  $\mathcal D$ by a simplex with the same vertices, while also adding all its subsets\footnote{This step is needed as if $S$ is a simplex in a abstract simplicial complex, all its subsets must also be in the complex.  },
that is, 
\[\mathcal A =  \{ S \subset \ver \m D\; |\; S \subset \ver F \textnormal{ for  } F\in \mathcal D \}.\]
Intuitively it is clear that $\m D$ and $\m A$ are homotopically equivalent. Indeed, to go from $\m D$ to $\geo \m A$, we simply need to expand some cells to higher dimension by converting them to simplices, which leaves their intersection structure unchanged.

To show this rigorously though, we will need to use the version of the nerve lemma presented in lemma~\ref{lem:nerve}.

\begin{lem}
    $D$  is homotopically equivalent to  $\m A$, that is,
    \[ D   \simeq  \m A .\]
    \label{lem:D_simeq_S}
\end{lem}
\begin{proof}
Recalling that the set of maximal cells cover a complex and using the nerve lemma we have
\[D = \bigcup \Max \m D\simeq \ner ( \Max \m D) \hspace{0.8cm}\textnormal{and}\hspace{0.8cm}\m A  \cong \geo \m A  = \bigcup \Max (\geo \m A) \simeq \ner (\Max (\geo \m A)).\]
The definition of $\m A$ induces a  bijection $f$ between the maximal cells of $\m D$ and the maximal cells of $\geo \m A$. Note that the intersection of cells is non-empty iff they have a common vertex, as cells intersect in faces. It is clear then that the intersection of maximal cells in  $\mathcal D$ will be non-empty precisely when the same is true for their image under $f$. That is, $f$ is a simplicial isomorphism between the two nerve complexes. Thus, we can conclude that $D \simeq \m A$.
\end{proof}

Now we want to exploit the simplicial structure of $\m A$ to show that it is contractible. We will do this by basically showing that there is a cover of $\{\vect \kappa> 0 \}$, $\mathcal C$, such that $\ner \m C$ is isomorphic to $\m A$ so that then by the nerve lemma, 
\[\m A \simeq \ner \m C \simeq \{\vect \kappa> 0\}\simeq \{*\},\]
where the cost set is contractible as it is convex. However, this approach does not quite work as to be able to apply the nerve lemma we need to work with compact sets. So instead we do the same argument with costs $\{\varepsilon \leq \vect \kappa\leq   m\}$ with $\varepsilon> 0$ small enough and $ m$ large enough such that there is at least one cost that attains each  element of $\m D$. 

\begin{thm}
    $D$ is contractible. 
    \label{thm:D_contractible}
\end{thm}
\begin{proof}
    By lemma~\ref{lem:D_simeq_S} it is enough to show that $\m A$ is contractible. Now for $v\in \ver \m A$ we define
    \[f(v):= \{\varepsilon \leq \vect \kappa\leq   m \; |\; v\in \amin(\vect \kappa)\}=  \left \{\varepsilon \leq \vect \kappa\leq  m \; |\; \sum_{ij\in \mathcal E_{\mathbf w}}\kappa_{ij}v_{ij} = \min_{\mathbf x \in P }\sum_{ij\in \mathcal E_{\mathbf w}}\kappa_{ij}x_{ij} \right\}.\]
    We note that this set is  bounded, closed and convex, as the functions in the equality are linear in $\vect \kappa$ and thus also continuous.
    Moreover, as every cell $\amin (\vect \kappa)$ contains a vertex,  $f(\ver \m A)$ covers $\{\varepsilon \leq \vect \kappa\leq m\}$,
    so by the nerve lemma we get, 
\[\ner f(\ver \m A) \simeq \bigcup f(\ver \m A) = \{\varepsilon \leq \vect \kappa\leq   m\} \simeq \{*\} . \]
So to finish the proof it is enough to show that that $f:\m A\rightarrow \ner f(\ver \m A)$ is a simplicial isomorphism. First note that by lemma~\ref{lem:F_polytope_complex}, we can find costs $\vect \kappa_v$ such that $\amin (\vect \kappa_v) = v$ for each vertex, which guarantees the bijectivity of  $f:\ver \m A \rightarrow f(\ver \m A)$. 
Now if $S\subset \ver \m A$ then, 
\begin{align*}
  S\in \m A & \iff \exists \vect \kappa,\;  \varepsilon\leq \vect \kappa\leq m \; \textnormal{ s.t. } S \subset \ver (\amin(\vect \kappa)) \\
  &\iff \exists \vect \kappa,\;  \varepsilon\leq \vect \kappa\leq  m \; \textnormal{ s.t. } \vect \kappa\in f(v) \;\; \forall v\in S \\
  &\iff \bigcap f(S) \neq \emptyset\iff f(S)\in \ner f(\ver \m A) ,
\end{align*}
so that $f$ is a simplicial isomorphism. 
\end{proof}

We note that the map $f$ in the proof is closely related to the concept of dual or polar of a polytope \cite{convex_polytopes}; however, we do not explore this further here. 

There are many approaches one can take to compute the space $D$ in practice. A simple method is to find first all the vertices of the polytope $P$ (see \cite{all_vertices} for a comparison between search methods). Then simply check which of the networks corresponding to these vertices are acyclic to find the vertices of $D$ and which convex combination of them are acyclic to find all the cells of $\m D$. This will however be a computationally intensive process due to the exponential number of combinations needed to be checked. 

\section{Minimal circularity}
\label{app:minimum_circularity}

In this section we show that finding minimum circularity is at least as hard as the Hamiltonian cycle problem. This is a well known NP-hard problem, even for sparse networks, and the currently best algorithms take exponential times in the number of nodes \cite{Hamiltonian_cycles}. Thus, the same will hold for the minimum circularity problem, once we show that it can be used to solve the Hamiltonian cycle problem.

First recall that the Hamiltonian cycle problem is to determine if an undirected unweighted graph $\mathbf g$ with $n$ nodes has a cycle visiting each node exactly once. Let $i$ be one of its nodes and consider then the directed weighted graph $\mathbf w$ that we obtain by replacing all edges in $\mathbf g$ by two directed edges of weight 1 in opposite directions, and with three extra nodes $i'$, $a$ and $b$. We take $i'$ to be a copy of $i$, that is, it is connected to the nodes $\mathcal N_{\mathbf g}\setminus \{i\}$ in the same way as $i$, and $a$, $b$ have the unique directed edges $ai$, $i'b$ also of weight 1. Now the directional parts of $\mathbf w$ have to simply send a unit of flow from $a$ to $b$ in an unrestricted manner while being acyclic. So it is clear that the maximal directional part (which corresponds to minimal circular one) will send all this flow through the longest path\footnote{That is, a sequence of distinct nodes where consecutive nodes are connected by following directed edges. } from $a$ to $b$. If this path goes through every node of $\mathbf w$, it will create a Hamiltonian cycle in $\mathbf g$ by merging $i$ and $i'$. 
Otherwise, the longest cycle containing $i$ in $\mathbf g$ does not contain all its nodes, so no Hamiltonian cycle exists. We can conclude that minimal circularity is a NP-hard problem.

\section{In-depth look at BFF algorithm}
\label{app:BFF}

\subsection{Background on Random walks}
\label{sec:random_walks}
Given a weighted directed network $\mathbf w$ we define the transition probabilities of a discrete time random walk in it by
\begin{equation}
   \bar{  w}_{ij} :=
\left\{
	\begin{array}{ll}
		 w_{ij}/ w_i^\outd  & \mbox{if }  w_i^\outd > 0 ,\\
		0 & \mbox{if }  w_i^\outd = 0 \mbox{ and } j\not = i,\\
        1 & \mbox{if }  w_i^\outd = 0 \mbox{ and } j= i.
	\end{array}
\right.
\label{eq:transition_matrix}
\end{equation}
This matrix defines the transition of the probability distributions $\mathbf a(t)$ between steps of a random walk  by 
\[\mathbf a(t+1)= \mathbf a(t) \bar {\mathbf w},\]
where $\mathbf a(0)$ represents the initial probability distribution on the nodes. Here by probability distribution we mean row vectors with non-negative entries adding up to 1. Of special interest are the \emph{stationary distributions} since the average of the process, $\frac{1}{t}\sum_{s=0}^{t-1}a(s)$, will always converge to one of them. They are characterized by being those  \emph{invariant} vectors, $\mathbf x$, that is 
\begin{equation} \mathbf x  \bar{\mathbf w} = \mathbf x , 
\label{eq:invariant} 
\end{equation}
that are also probability distributions, i.e.~ $\sum_{i\in \mathcal N_\mathbf w} x_i = 1$ and $x_{i} \geq 0$. 

A walk in a network  is a sequence of nodes where consecutive nodes are connected by following directed edges.  We say that a network is \emph{strongly connected} if there exist walks in both directions between any pair of nodes. Such networks have a unique stationary distribution, which in turn  has all positive entries.  For a general network $\mathbf w$,  having a walk in both directions defines an equivalence relation on its nodes, so that we can decompose $\mathbf w$ into its SCCs, or for us simply \emph{components}. A component is called \emph{absorbing} if it has no outgoing edges to other components. Then the unique stationary distribution of an absorbing component gives us an invariant vector of $\bar {\mathbf w}$, by setting all entries outside the component to zero. These vectors give a basis for the space of invariant vectors satisfying \eqref{eq:invariant}.

\subsection{Alternative motivation for BFF and proofs}
Suppose that we are simply interested in finding some circular subnetwork $\tilde {\mathbf c}$ of $\mathbf w$, such that it preserves the edge proportions as in \eqref{eq:ratio}. Then,  $\tilde c_{ij} = c_i^\outd \bar{w}_{ij}$ so that $\tilde {\mathbf c}$ is completely determined by the out-weight row vector. Moreover, the computations in \eqref{eq:a^in} done for $\tilde {\mathbf c}$ instead of $\mathbf a (t)$ yield $\tilde {\mathbf c}^\ind = \tilde {\mathbf c}^\outd \bar{\mathbf w}$. 
Thus, the condition of being circular or balanced, $ \tilde {\mathbf c}^\ind = \tilde {\mathbf c}^\outd$, is precisely, 
\[ \tilde {\mathbf c}^\outd \bar{\mathbf w} = \tilde {\mathbf c}^\outd\]
i.e. $\tilde {\mathbf c}^\outd$ is invariant under $\bar{\mathbf w}$. 

In principle we can choose any invariant vector with $0\leq \tilde {\mathbf c}^\outd \leq \mathbf w^\outd$. However, if in algorithm~\ref{alg:Comprsession_BFF} we choose $\tilde {\mathbf c}$ in this manner, instead of using  \texttt{Maximal\_Invariant},
 we need the remainder flow $ \mathbf w - \tilde {\mathbf c} $ to have a sink node to guarantee convergence. We can achieve this by taking $\tilde {\mathbf c}^\outd$ to be a scaling of any stationary distribution, such that  ${\mathbf c}^\outd \leq \mathbf w^\outd$ has a binding\footnote{That is, $\mathbf c^\outd\leq \mathbf w^\outd$ and for some node $ c^\outd_i =  w^\outd_i $. } inequality.
Importantly, despite having freedom when choosing $\tilde {\mathbf c} $ in each iteration  the final cumulative circular flow, $\mathbf c$, will always be the same, as we proceed to show. 

First note that each  cycle is contained in a single component, so that all circularity flows within components and not between them. Let $\mathcal C$ be the nodes in a component of $\mathbf w$ and denote by ${\mathbf w}_\mathcal{C}$ its restriction to these nodes. Then, for any invariant vector of $\bar{\mathbf w}$, its restriction to $\mathcal C$, $\tilde {\mathbf c}_{\mathcal C}^\outd $, will be an invariant vector of $\bar {\mathbf w}_\mathcal{C}$. Recall that in a component there is a unique invariant stationary distribution, $\mathbf s$, so $\tilde {\mathbf c}_\mathcal{C}^\outd$ must be a scaling of it. We show that algorithm~\ref{alg:Comprsession_BFF} will eventually remove the only scaling of $\mathbf s$ with a binding inequality $\lambda \mathbf s\leq  {\mathbf w}_\mathcal{C}^\outd$, regardless of the method used to choose $\tilde {\mathbf c}$. Assume that this does not occur in the first iteration so $\tilde {\mathbf c}_\mathcal{C}^\outd< {\mathbf w}_\mathcal{C}^\outd$, or equivalently that we do not create sink nodes in $\mathcal C$ in this iteration. Then the stationary distribution in $\mathcal C$ in the next iteration of the algorithm will be unchanged since for edges $ij$ in it, 
\[[\overbar{\mathbf w -\tilde{\mathbf c}}]_{ij} = \frac{w_{ij}-\tilde{c}_{ij}}{w_i^\outd - \tilde{c}_i^\outd} = \frac{w_{ij}-\tilde{c}_i^\outd \bar w_{ij}}{w_i^\outd - \tilde{c}_i^\outd} =\frac{w_i^\outd w_{ij}-\tilde{c}_i^\outd  w_{ij}}{w_i^\outd(w_i^\outd - \tilde{c}_i^\outd)} = \bar{w}_{ij}, \]
as $c^\outd_i \not = w^\outd_i$.
This will continue until at some  iteration we create a sink\footnote{Since we know the algorithm terminates this must occur at some iteration.} in $\mathcal C$, at which point all the out flows removed are scalings of $\mathbf s$ which add up precisely to $\lambda \mathbf s$.

In conclusion for each component we will eventually remove the unique flow coming from the maximal invariant vector of it. Once this flow is removed for a particular component the removal of sinks in the next iteration will  reduce its  size and  may even subdivide it. The resulting components undergo then the same process. 

Besides showing that the final convergence is independent of the choices of $\tilde {\mathbf c}$ the above argument hints at a recursive implementation of the BFF algorithm, which is described in  algorithm~\ref{alg:BFF_recursive}. As in algorithm~\ref{alg:Comprsession_BFF}, when adding networks of different sizes, we simply add the weights of matching edges, where the node labels are kept from the original network $\mathbf w$.
\begin{algorithm}[hbt!]
\caption{Recursive implementation of BFF algorithm, returning the circular part.  }
\label{alg:BFF_recursive}
 \DontPrintSemicolon
\SetKwFunction{Fbff}{BFF}
\SetKwFunction{Fcomp}{Components}
\SetKwFunction{Fmaxinv}{Maximal\_Invariant}
\SetKwProg{Pn}{Function}{:}{\KwRet $c$}
  \Pn{\Fbff{$w$}}{
$c = 0$\\
\If{$w$ has an edge}{
    \For{component $\mathcal C$ of $w$}{
    $\tilde c =$ \Fmaxinv{$w_{\mathcal C}$}\\
    $c = c + \tilde {c} $ $+$ \Fbff{$w_{\mathcal C}-\tilde c$}\\
    }
}
}
\end{algorithm}

A big advantage of this approach is that  $\mathbf w_{\mathcal C}$ has a unique invariant vector (up to scaling) and thus  finding  \texttt{Maximal\_Invariant} via standard linear algebra techniques is straight-forward and fast.  Moreover, the recursivity of the method quickly reduces the dimensionality of it. All this makes this approach computationally more efficient in practice and it is the one we roughly take in our code.  

To compute its complexity, note that finding the stationary distribution to compute \texttt{Maximal\_}\allowbreak\texttt{Invariant} is the bottle-neck. Indeed, one can  find the components in linear time on the number of nodes and edges \cite{Tarjan_1972},
whereas finding the stationary distribution is  nearly linear  \cite{time_stationary_dist_2018} for approximate solution. There are exact methods that take $\tilde O (m^{3/4}n +m n ^{2/3})$ where $n =|\mathcal N_{\mathbf w}|$, $m = |\mathcal E_{\mathbf w}| $ and the $\tilde O$ notation suppresses polylog factors, see \cite{time_stationary_dist_2016}. Now, each time we call \texttt{Maximal\_Invariant} we create a sink node, which in the next iteration will be part of a singleton component. 
So we call this function on non-singleton networks at most as many times as there are nodes and thus algorithm~\ref{alg:BFF_recursive} takes $\tilde O (m^{3/4}n^2 +m n ^{5/3})$ and in particular polynomial time. 

From algorithm~\ref{alg:BFF_recursive}, we can also deduce that $\mathbf c$ contains all cycles of $\mathbf w$. Indeed, each cycle will be contained in a component $\mathcal C$. Then, as $\bar{\mathbf w}_{\mathcal C }$ is strongly connected its stationary distribution  will have positive weight on all its nodes. Thus, \texttt{Maximal\_Invariant(}${\mathbf w}_{\mathcal C }$\texttt{)} will have the same edges as ${\mathbf w}_{\mathcal C }$ but with different weights and in particular they will contain the same cycles. 

To finish this section we prove theorem~\ref{thm:maximal_invariant}.

\begin{proof}[Proof of theorem~\ref{thm:maximal_invariant}]
As the function defining the dynamics of $\mathbf a^\outd(t)$ is continuous we know that its limit is a fixed point, that is
    \[\min (\tilde {\mathbf c}^{\outd} \bar {\mathbf w}, \tilde{\mathbf c }^\outd)= \tilde {\mathbf c}^\outd  \]
    where the minimum is taken entrywise. 
    It follows that $\tilde {\mathbf c}^{\outd} \bar {\mathbf w} \geq \tilde {\mathbf c}^\outd $,  and as $\bar {\mathbf w}$ is a stochastic matrix, i.e. its rows add up to 1, we have,
    \[\sum_{i\in \mathcal N_{\mathbf w}} [\tilde {\mathbf c}^{\outd} \bar {\mathbf w}]_i = \sum_{i\in \mathcal N_{\mathbf w}}\sum_{j\in \mathcal N_{\mathbf w}}  \tilde  c^{\outd}_j\bar { w}_{ji} = \sum_{j\in\mathcal N_{\mathbf w}}  \tilde  c^{\outd}_j \sum_{i\in \mathcal N_{\mathbf w}}\bar { w}_{ji} =\sum_{i\in \mathcal N_{\mathbf w}} \tilde c_i^\outd. \]
     We conclude that $ \tilde {\mathbf c}^{\outd} \bar {\mathbf w} = \tilde {\mathbf c}^\outd $ so it is invariant. 

    To show that $\mathbf c ^\outd$ is maximal we first prove that $0\leq \mathbf a ^\outd(t)-\mathbf x$ by  induction. First note that  $\mathbf x\leq \mathbf w^{\outd} =  \mathbf a^{\outd}(0)$. Also note that from the minimum in the dynamics, either $a_i ^\outd(t+1)-x_i =  a_i ^\outd(t)-x_i$ which is non-negative by induction, or
    \[ a_i ^\outd(t+1)-x_i = [\mathbf a ^\outd(t) \bar{\mathbf w}]_i -x_i = [(\mathbf a ^\outd(t)  -\mathbf x)\bar{\mathbf w}]_i\geq 0, \]
    where we used that $\mathbf x$ is invariant, the induction hypothesis and $\mathbf {\bar w}\geq 0$.
    Now taking limits we obtain $\mathbf x\leq \tilde {\mathbf c }^\outd$ as desired.
\end{proof}

\section{Random networks}
\label{app:Random_networks}

There are many ways to generate random networks but perhaps the most widely used one is the Erdős–Rényi model. This term is used for two closely related ensembles of undirected unweighted networks, here we present their directed unweighted versions. The \emph{uniform} model $G_{n,m}$, generates random unweighted networks with $n$ nodes and $m$ directed edges, chosen uniformly in the set of all possible edges $\mathcal L = \{ij \; | \; i\not = j \textnormal { and } 1\leq i, j\leq n\}$. The \emph{binomial} model $G^{n,p}$, generates random unweighted networks with $n$ nodes, where  each possible directed edge $ij\in \mathcal L$ is included independently with probability $p$. One is usually interested in the expected properties of these networks when $n$ tends to infinity, see \cite{Random_graphs}.

We propose a slight modification to the uniform model, to obtain weighted networks. In $W_{n,w}$ networks of $n$ nodes are generated by uniformly sampling with replacement $w$ edges in $\mathcal L$. The weight of each edge is then given by the number of times they have come up in the sampling, where  weight zero represents the absence of an edge.   Thus, a generated network may have less than $w$ edges but the sum of its weights is always $w$. 

First we show that for sparse networks $W_{n,w}$ is very similar to $G_{n,w}$, as we expect few edges with weight larger than 1. To this end note that $|\mathcal L| = n(n-1)$ which we denote by $l$ and use throughout this section.

\begin{lem}
Let $Y$ be the sum of the edge weights bigger than 1 in $W_{n,w}$, and assume that $w = kn$ with $k>0$ fixed. Then $\mathbb EY$ is asymptotically bounded above by $k^2$. 
\label{lem:total_weight>1}
\end{lem}
\begin{proof}
Sampling the edges sequentially, the probability that in the $i$th sample we choose an edge already present is at most $(i-1)/l$, which would increase $Y$ by 2 or 1. Thus, 
\[\mathbb E Y \leq 2\sum_{i = 1}^w \frac{i-1}{l} = \frac{nk(nk-1)}{n(n-1)} \longrightarrow k^2  \]
as $n\to \infty$, where we have used the arithmetic sum formula. 
\end{proof}

Note that in the lemma's setting the total weight $w$ tends to infinity in contrast to $\mathbb E Y$. 

Staying in the sparse regime we study how many short cycles we expect. Let $X_i$ be the number of distinct cycles of length $i$ and as we will move between probability spaces we denote by $X_i(W_{n,w})$ the corresponding random variable in the probability space $W_{n,w}$, similarly for $G_{n,m}$ and $G^{n,p}$. 
It is easiest to count cycles in $G^{n,p}$ since networks formed by subsets of its nodes also behave as a binomial ensemble, in particular, 
\[ \mathbb E X_i(G^{n,p}) =\binom{n}{i} \mathbb E X_i(G^{i,p}). \]
Note that cycles of length $i$ in a network with $i$ nodes are simply permutations of the nodes where the starting node is arbitrary so there are $(i-1)!$ possible cycles. Then, as the probability of all the edges of a cycle being present is $p^i$ we have $\mathbb E X_i(G^{i,p}) = (i-1)! p^i$. In conclusion, 
\begin{equation}
\mathbb E X_i(G^{n,p}) = \frac{n (n-1) \cdots (n-i+1)}{i}p^i\leq np\frac{((n-1)p)^{i-1}}{i} \leq  2\frac{k^i}{i}
\label{eq:EX_i(G)}
\end{equation}
where in the last equality we have assumed that we are in the sparse regime $p  = k/(n-1)$ for $k>0$ fixed. 
Note that in this regime the expected number of edges is $m = kn$ and the expected average in/out-degree is $ k$.
As $G^{n,p}$ is closely related to $G_{n,m}$ one can show that  $\mathbb E X_i(G_{n, m})$ is asymptotically bounded above by $2k^i/i$ for any fixed $i$. 
We do not show the details here as we will do exactly the same argument applied to circularity in the proof of the next proposition. Finally note that repeating edges cannot increase $X_i$ so that $\mathbb E X_i(W_{n, m})\leq \mathbb E X_i(G_{n, m})$ and thus is also asymptotically bounded by $2k^i/i$. 

We now show that in very sparse networks we expected zero circularity. We will need the following technical result, see \cite[Corollary 2.3]{Random_graphs}. 

\begin{lem}[Chernoff’s bound]
If $X$ be a binomial random variable and $0<\varepsilon\leq 3/2$, then
\[\mathbb P\left ( |X-\mathbb E X| \geq \varepsilon \mathbb E X \right ) \leq 2 \exp   \left (- \frac{\varepsilon^2 \mathbb E X}{3}\right ).\]
\label{lem:chrenoff}
\end{lem}

\begin{prop}
Let $C$ be the maximal circularity and $w = kn$ with $0<k<1$ fixed. Then, $\mathbb E C(W_{n, w})\to 0$ as $n\to \infty$. 
\end{prop}
\begin{proof}
Note that the unnormalised circularity, $wC$, is bounded above by  the sum of the weights of edges that are contained in cycles, which we denote by  $Z$. Then, if we let $Y$ be the sum of the edges weights bigger than 1, we have
\begin{equation}
\mathbb E C(W_{n, w}) \leq \frac{\mathbb E Z(W_{n, w})}{w} \leq  \frac{\mathbb E Z(G_{n, w})}{w}  + \frac{\mathbb E Y(W_{n, w})}{w},
\label{eq:EC_max}
\end{equation}
where  $Z$ is computed on unweighted networks by assuming their weights are 1.  
Now by lemma~\ref{lem:total_weight>1} the term on the right tends to 0 as $n\to \infty$ so we focus on bounding $\mathbb E Z(G_{n, w})$. Consider now  $X = \sum_{i=2}^n i X_i $, where $X_i$ is the number of cycles of length $i$. Clearly $ Z(G_{n, w})\leq X (G_{n, w})$ 
and as before it is easier to study $X (G^{n, p})$ with $p= \frac{w}{l} =\frac{k}{n-1}$. We have, 
\begin{equation}
\mathbb E X(G^{n, p}) =  \sum_{i=2}^n i\mathbb  E X_i(G^{n, p}) \leq \sum_{i=2}^n 2k^i \leq \frac{2 k^2}{1-k}
\label{eq:EX(G)}
\end{equation}
where we have used \eqref{eq:EX_i(G)} and $0<k<1$. 

We now convert the result from  $G^{n, p}$ to $G_{n, w}$. To this end we note that in general if $G$ a network with distribution $G^{n, q}$, $\mathcal P$ is a network property and denoting by $E$ the number of edges we have  
\[\mathbb P  \big (G \in \mathcal P \;|\; E(G ) = m \big ) = \mathbb P (G_{n,m} \in \mathcal P). \]
Then, by expanding with conditional expectation we get
\begin{equation}
\mathbb E X(G^{n, p}) = \sum_{m= 0}^l \mathbb P \big (E(G^{n, p})= m \big ) \mathbb E X(G_{n,m}).
\label{eq:sum_P(E=m)}
\end{equation} 
Now note that $E (G^{n,p})$ is a binomial random variable with $l$ experiments and probability $p$. Thus, $\mathbb E E (G^{n,p}) = lp = w$ and given $0<\varepsilon<1$ fixed we have  from lemma~\ref{lem:chrenoff}
\[ \mathbb P\left ( |E (G^{n,p})-w| < \varepsilon w \right ) \geq 1- 2 \exp   \left (- \frac{\varepsilon^2 w}{3}\right ) \longrightarrow 1\]
as $n\to \infty$, since then $w\to \infty$. Thus, given  $0<\lambda <1$ fixed, for all $n$ large enough
\[\mathbb E X(G^{n, p}) \geq \lambda  \mathbb E X(G_{n,(1-\varepsilon)w}),\]
where we have used \eqref{eq:sum_P(E=m)} and that $X$ can only increase when adding edges. 
Note that  $\tilde k = \frac{k}{1-\varepsilon}<1$ if $\varepsilon$ is chosen small enough and then for $n$ large enough  
\[\mathbb E X (G_{n,w}) \leq \frac{1}{\lambda}\mathbb E X\Big (G^{n,\frac{\tilde k}{n-1}} \Big ) \leq \frac{2\tilde k^2}{\lambda(1-\tilde k)}  =  \frac{2k^2}{(1-\varepsilon)\lambda(1-\varepsilon-k)} \]
where in the second inequality we use \eqref{eq:EX(G)}.
Although it is needed here, note that as $\varepsilon$ and $1-\lambda$ can be chosen arbitrarily small, we obtain the asymptotic bound $2k^2/(1-k)$. As this value is finite we can conclude from \eqref{eq:EC_max} that the circularity converges to 0 as  $n\to \infty$.  
\end{proof}

We now find that for a dense enough regime the network circularity will tend to 1. The basic idea is that by the law of large numbers  all possible edges will have roughly the same large weight. There will be some variability between edge weights but  most of the network's weight will be contained in a (balanced) complete graph. Formally we have

\begin{prop}
Let $C$ be the minimal circularity on $W_{n,w}$ with $w = f(n)n^2 log(n)$ for any positive function $f$ such that $f(n) \to \infty$ when $n\to \infty$. Then, $\mathbb E C\to 1$ as $n\to \infty$.   
\end{prop}
\begin{proof}
Note that if   $\mathbf w$ has  $a = \min_{ij\in \mathcal L } w_{ij}\not = 0$ then all circular parts will consist of a complete graph with all weight equal or larger to $a$.
Thus, if we let $X_{ij}$ be the weight of the edge $ij\in \mathcal L$ on $W_{n,w}$, we have 
\begin{equation}
\begin{split}
\mathbb E C &\geq  \frac{l}{w} \mathbb E \left ( \min_{ij\in \mathcal L} X_{ij} \right ) \geq \frac{l \lambda}{w} \mathbb P \left ( \min_{ij\in \mathcal L} X_{ij}\geq \lambda \right) = \frac{l \lambda}{w} \big (1- \mathbb P \left (  X_{ij}< \lambda \textnormal{ for some }  ij\in \mathcal L \; \right)\big )\\
&\geq \frac{l \lambda}{w} \big (1-l \mathbb P \left (  X_{12}< \lambda \right)\big ),
\end{split}
\label{eq:ECmin}
\end{equation}
for arbitrary  $\lambda$. Note that $X_{12}$ follows a binomial distribution with $w$ samples and probability $1/l$. Let $\varepsilon = \sqrt{l/w} \log^{1/2}(l^6)$ so $\varepsilon^2 = \frac{6l\log(l)}{w} \leq \frac{12}{f(n)} \to 0$ as $n\to \infty$ and we can apply lemma~\ref{lem:chrenoff} to $X_{12}$ for all $n$ large enough. 
Taking $\lambda = (1-\varepsilon) \mathbb E X_{12} $, where $\mathbb E X_{12}= w/l$, we have
\[ \mathbb P(X_{12}<\lambda ) \leq \mathbb P \left ( |X_{12}- \mathbb E X_{12}| \geq \varepsilon\mathbb E X_{12} \right ) \leq 2 \exp \left (\frac{-\varepsilon^2 w}{3 l}\right ) = \frac{2}{l^2}, \] 
for all $n$ large enough.
Then, from \eqref{eq:ECmin} we obtain 
\[\mathbb E C\geq \frac{l }{w} \frac{(1-\varepsilon)w}{l} \left (1-l \frac{2}{l^2}\right) = (1-\varepsilon)\left (1-\frac{2}{l} \right ) \to 1\]
as $n\to \infty$.

\end{proof}

We note that the equivalent result for the unweighted models are trivial as if $m = l $ and $p = 1$ then  $G_{n,m}$ and $G^{n,p}$ are always complete unweighted graphs and thus completely circular. One could however try to improve upon these and previous bounds.  

\end{document}